\documentclass{llncs}
\usepackage{amsmath,amssymb,textcomp}
\usepackage{graphicx}
\usepackage{tikz,ifthen}
\usetikzlibrary{shapes,arrows,decorations.pathreplacing,chains,matrix,positioning,scopes,calc,fit}

\newtheorem{observation}{Observation}

\DeclareMathOperator{\pref}{prefix}
\DeclareMathOperator{\suff}{suffix}

\DeclareMathOperator{\mirror}{R}

\DeclareMathOperator{\bin }{bin}
\DeclareMathOperator{\pad }{pad}
\DeclareMathOperator{\chain }{chain}

\newcommand{\NP}{\textbf{NP}}
\newcommand{\hide}[1]{}

\newcommand{\dollar}{\ensuremath{\$}}
\newcommand{\var}{\ensuremath{\hat{x}}}
\newcommand{\lit}{\ensuremath{\hat{c}}}
\newcommand{\underscore}[4][black] {
  \draw[#1,very thick] (#2,#3) ++(-.4,-.3) -- +(0,.1);
  \draw[#1,very thick] (#2,#3) +(-.4,-.3) -- ++(#4.4,-.3) -- +(0,.1);
}
\newcommand{\overscore}[4][black] {
  \draw[#1,very thick] (#2,#3) ++(-.4,+.3) -- +(0,-.1);
  \draw[#1,very thick] (#2,#3) +(-.4,+.3) -- ++(#4.4,+.3) -- +(0,-.1);
}

\begin{document}
\mainmatter              

\title{The complexity of string partitioning}
\titlerunning{Complexity of string partitioning}
\author{Anne Condon\inst{1} \and J\'{a}n Ma\v{n}uch\inst{1,2} \and Chris Thachuk\inst{1}
}
\authorrunning{Condon, Ma\v{n}uch and Thachuk}
\institute{
  Dept. of Computer Science, University of British Columbia, Vancouver
  BC, Canada, \email{\{condon,jmanuch,cthachuk\}@cs.ubc.ca}
  \and
  Dept. of Mathematics, Simon Fraser University, Burnaby BC, Canada
}

\maketitle

\begin{abstract}
Given a string $w$ over a finite alphabet $\Sigma$ and an integer $K$, can
$w$ be partitioned into strings of length at most $K$,
such that there are no \emph{collisions}?  We refer to this question
as the \emph{string partition} problem and show it is \NP-complete
for various definitions of collision and for a number of interesting
restrictions including $|\Sigma|=2$.  This establishes the hardness of
an important problem in contemporary synthetic biology, namely, oligo
design for gene synthesis.
\end{abstract}

\section{Introduction}
\label{sec:introduction}

Many problems in genomics have been solved by the application of
elegant polynomial-time string algorithms, while others amount to solving
known \NP-complete problems; for instance, sequence assembly amounts to
solving \emph{shortest common superstring} \cite{Karp1993}, and
genome rearrangement to \emph{sorting strings by reversals and
  transpositions} \cite{ChrIrv2001}.  The hardness of
these problems has motivated extensive research into heuristic
algorithms as well as polynomial-time algorithms for useful 
restrictions
\cite{Eriksen2002,Hartman2003,YanEtAl2005,Hannenhalli1996,GolKolZhe2005,MyersEtAl2000,PevTanWat2001}.
In a similar vein, we establish the hardness of the following
fundamental question: can a string be partitioned into factors 
(\textit{i.e.} substrings), of
bounded length, such that no two \emph{collide}?  We refer to this
as the \emph{string partition} problem and study it under various
restrictions and definitions of what it means for two factors to
\emph{collide}.

The study of string partitioning is motivated by an increasingly
important problem arising in contemporary synthetic biology, namely
gene synthesis.  This technology is emerging as an important tool
for a number of purposes including the determination of RNAi targeting
specificity of a particular gene \cite{KumGusKle2006}, design of novel proteins
\cite{CoxEtAl2007} and the construction of complete bacterial genomes
\cite{GibEtAl2008}. There have been numerous studies utilizing
synthetic genes to determine the potential of gene vaccines 
\cite{LinEtAl2006,CidEtAl2003,RodWu2003,ChlSchSan2005}.
Despite the tremendous need for synthetic genes for both
interrogative studies and for therapeutics, construction of genes, or
any long DNA or RNA sequence, is not a trivial matter. Current
technology can only produce short oligonucleotides (oligos)
accurately.  As such, a common approach is to design a set of
oligos that could assemble into the desired sequence
\cite{SteEtAl1995}.  

To understand the connection between string partitioning and gene
synthesis, consider the following.  A DNA \emph{oligo}, or
\emph{strand} is a string over the four letter alphabet
$\{\mathtt{A,C,G,T}\}$.  The \emph{reverse complement} $F'$ of an
oligo $F$ is determined from $F$ by replacing each $A$ with a $T$ and
vice versa, each $C$ with a $G$ and vice versa, and reversing the
resulting string.  Two DNA oligos $F$ and $F'$ are said to
\emph{hybridize} if a sufficiently long factor of $F$ is the reverse
complement of a factor of $F'$ (see Figure
\ref{fig:foiled-oligo-design}).
A DNA \emph{duplex} consists of a positive strand
and its reverse complement, the negative strand.  The
\emph{collision-aware oligo design for gene synthesis} (CA-ODGS)
problem is to determine cut points in the positive and negative
strands, which demarcate the oligos to be synthesized, such that the
resulting design will successfully self-assemble.  For the oligos to
self-assemble correctly, they should 1) alternate between the positive
and negative strands, with some overlap between successive oligos, and
2) only hybridize to 
the oligos they overlap with by design.
Since there is variability in
the length of the selected oligos, there are exponentially many
designs.

\begin{figure}
  \begin{center}
  \includegraphics[width=.85\textwidth]{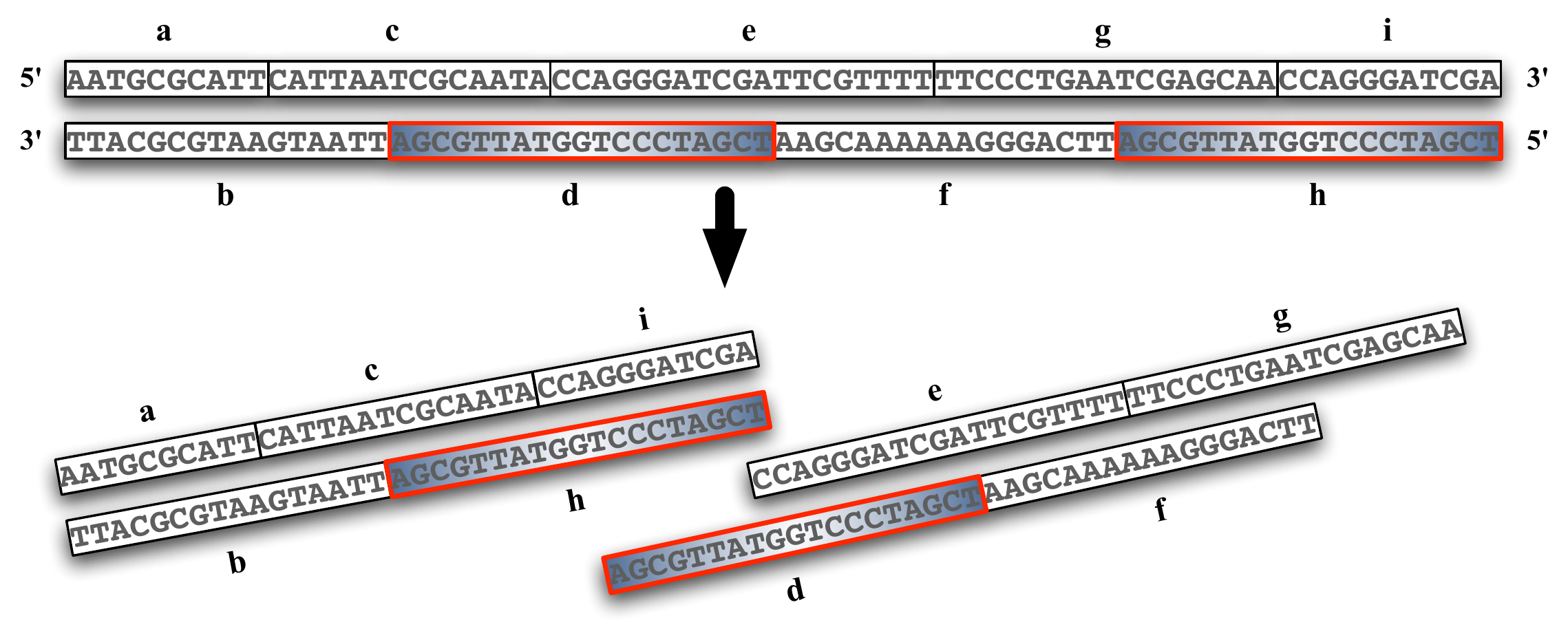}
  \caption{An intended self-assembly (top) of a set of oligos for a
    desired DNA duplex.  A foiled self-assembly (bottom) of the same
    oligos due to $d$ and $h$ being identical.}
  \label{fig:foiled-oligo-design}
  \end{center}
\end{figure}

In previous work \cite{ConManTha2008}, the authors provided some
evidence that the CA-ODGS
problem may be hard by showing that partitioning a string into factors, of
bounded length, such that no two are equal is \NP-complete, even for
strings over a quaternary alphabet.  See Figure
\ref{fig:foiled-oligo-design} for an example design that assembles
incorrectly into two fragments, with the wrong ordering of oligos and
therefore primary sequence, due to identical oligos.
In this work, we study the
underlying string partition problem in much greater detail.  We show that partitioning strings
such that no selected string is a copy/factor/prefix/suffix of another is
\NP-complete.  We begin by showing that the more general problem of
partitioning a set of strings is hard and then we show how
those instances can be reduced to single string instances,
for each respective definition of collision.  See Figure
\ref{fig:string-partition} for an example of a single string instance
(left) and set of strings instance (right).  In all cases, we
demonstrate the problems remain hard even when restricted to binary strings.

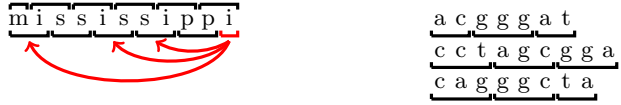
\begin{figure}


  \centering
  \small
  \begin{tikzpicture}[xscale=0.4,scale=0.70]
    \def\AS{1.0}
    
    \overscore[black]{1}{\AS}{0}
    \overscore[black]{2}{\AS}{1}
    \overscore[black]{4}{\AS}{1}
    \overscore[black]{6}{\AS}{1}
    \overscore[black]{8}{\AS}{1}
    \overscore[black]{10}{\AS}{1}
    
    \underscore[black]{1}{\AS}{1}
    \underscore[black]{3}{\AS}{1}
    \underscore[black]{5}{\AS}{1}
    \underscore[black]{7}{\AS}{1}
    \underscore[black]{9}{\AS}{1}
    \underscore[red]{11}{\AS}{0}
    
    \draw[anchor=mid] (1 ,\AS) node{m};
    \draw[anchor=mid] (2 ,\AS) node{i};
    \draw[anchor=mid] (3 ,\AS) node{s};
    \draw[anchor=mid] (4 ,\AS) node{s};
    \draw[anchor=mid] (5 ,\AS) node{i};
    \draw[anchor=mid] (6 ,\AS) node{s};
    \draw[anchor=mid] (7 ,\AS) node{s};
    \draw[anchor=mid] (8 ,\AS) node{i};
    \draw[anchor=mid] (9 ,\AS) node{p};
    \draw[anchor=mid] (10,\AS) node{p};
    \draw[anchor=mid] (11,\AS) node{i};
    
    \draw[very thick,red,->] (11,0.6) .. controls +(-1,-.25) and +(1,-.25)
    .. (7.5,0.6);
    \draw[very thick,red,->] (11,0.6) .. controls +(-1,-.5) and +(1,-.5)
    .. (5.5,0.6);
    \draw[very thick,red,->] (11,0.6) .. controls +(-1,-1) and +(1,-1)
    .. (1.5,0.6);

    \begin{scope}[shift={(20,0)}]
      \foreach \x/\l in {1/1,3/2,6/1}{
        \underscore[black]{\x}{1}{\l}
      }
      
      \foreach \x/\l in {1/a,2/c,3/g,4/g,5/g,6/a,7/t}{
        \draw[anchor=mid] (\x,1) node{\l};
      }      

      \foreach \x/\l in {1/2,4/2,7/2}{
        \underscore[black]{\x}{.4}{\l}
      }
      
      \foreach \x/\l in {1/c,2/c,3/t,4/a,5/g,6/c,7/g,8/g,9/a}{
        \draw[anchor=mid] (\x,.4) node{\l};
      }      

      \foreach \x/\l in {1/2,4/2,7/1}{
        \underscore[black]{\x}{-.2}{\l}
      }
      
      \foreach \x/\l in {1/c,2/a,3/g,4/g,5/g,6/c,7/t,8/a}{
        \draw[anchor=mid] (\x,-.2) node{\l};
      }      
    \end{scope}
    
  \end{tikzpicture}

  \centering
  \caption{(Left) Two partitions are shown for the string \emph{mississippi}.
    The selected strings in both partitions have maximum length 2.  The
    partition shown above the string is factor-free: no selected string is
    a factor of another; however, the partition shown below the string
    is not factor-free. (Right) A valid factor-free multiple string partition of a
    set of three strings into selected strings of maximum length 3.}
  \label{fig:string-partition}
\end{figure}

\section{Preliminaries}\label{sec:prelim}

A \emph{string} $w$ is a sequence of letters over an alphabet
$\Sigma$.  Let $|w|$ denote the length of $w$, $w^{\mirror}$ a
mirror image (reversal) of $w$, and let $(w)^i$ denote
the string $w$ repeated $i$ times. The empty string is denoted as
$\varepsilon$. String $x$ is a \emph{factor} of $w$ if
$w=\alpha x \beta$, for some (possibly empty) strings $\alpha$ and
$\beta$.  Similarly, $x$ is a \emph{prefix} (\emph{suffix}) of $w$
if $w=x \beta$ ($w=\alpha x$) for some (possibly empty) strings
$\alpha$ and $\beta$. The prefix (suffix) of length $k$ of $w$ will be
denoted as $\pref_{k} (w)$ ($\suff_{k} (w)$).

A $K$-\emph{partition} of $w$ is a sequence $P = p_1,p_2, \ldots
,p_l$, for some $l$, where each $p_i$ is a string over $\Sigma$ of
length at most $K$ and $w = p_1 p_2 \ldots p_l$.
We say that strings $p_{1},\dots,p_{l}$ are \emph{selected} in the
$K$-partition and that strings $p_{i}\dots p_{j}$, $1\le i\le j\le l$,
are \emph{super-selected}, with respect to the selected strings.
We say $P$ is \emph{equality-free}, \emph{prefix-free}, \emph{suffix-free}, or
\emph{factor-free} if for all $i,j$, $1 \le i \neq j \le l$, neither
$p_i$ nor $p_j$ is a copy, prefix, suffix, or factor, respectively, of the
other.  We say such partitions are \emph{valid} (for the problem in
question); otherwise, we say the partition contains a
\emph{collision}.  
We generalize the notion of a $K$-partition to a set of strings
$\mathcal{W}$ to mean a $K$-partition for each string in
$\mathcal{W}$.  The length of $\mathcal{W}$ is the combined length of
the strings in the set and will be denoted by $||\mathcal{W}||$. 
A $K$-partition for a set of strings is valid if no
two elements in any, possibly different, partition collide.  Finally,
we will refer to the boundaries of a partition of string $w$ as
\emph{cut points}, where the first cut point 0 and the last cut
point $|w|$ are called trivial. For instance, the first partition of
mississippi in Figure~\ref{fig:string-partition} has the following
non-trivial cut points $1,3,5,7$ and $9$.

In what follows we will prove \NP-completeness of various string
partitioning problems by showing a polynomial reduction from an
arbitrary instance of 3SAT(3), a problem shown to be \NP-complete by
Papadimitriou \cite{Papadimitriou1994}.

\begin{problem}[3SAT(3)]
  \ \\
  \noindent\emph{Instance}: A formula $\phi$ with a set $C$ of clauses
  over a set $X$ of variables in conjunctive normal form such that:
  \begin{enumerate}
  \item every clause contains two or three literals,
  \item each variable occurs in exactly three clauses, once negated
    and twice positive.
  \end{enumerate}

  \noindent\emph{Question}: Is $\phi$ satisfiable?
\end{problem}

\section{The String Partition Problems}
\label{sec:string-part-probl}

\begin{figure}
\centering
\begin{tikzpicture}[font=\footnotesize,
    level distance=33mm,
    level 2/.style={sibling distance=6mm},
    level 1/.style={sibling distance=14mm},
    pnode/.style={font=\footnotesize},
    enode/.style={font=\footnotesize},
    grow=right
  ]

  \node[pnode] {3SAT(3)}
     [thick,dotted,black]
     child{ node[pnode] {EF-MSP(K=2)} 
       child{ node[pnode] {EF-MSP(L=2)} 
         child{ node[pnode] {EF-SP(L=2)} }
       }
       child{ node[pnode] {EF-SP(K=2)} }
     }
     child{ node[pnode] {FF-MSP(K=3)} 
       child{ node[pnode] {FF-MSP(L=2)}
         child{ node[pnode] {FF-SP(L=2)} }
       }
       child{ node[pnode] {FF-SP(K=3)} }
     }
     child{ node[pnode] {PF-MSP(K=2)} 
       child{ node[pnode] {PF-SP(K=2)} }
       child{ node[pnode] {PF-MSP(L=2)}
         child{ node[pnode] {PF-SP(L=2)} }
       }
     }
     ;
\end{tikzpicture}

  \caption{Chain of reductions for different string partition
    variations from original 3SAT(3) problem.  $K$ is maximum selected
  string size and $L$ is maximum alphabet size.  Parameters are
  unbounded if not shown.  $EF$, $FF$ and $PF$ are equality-free,
  factor-free, and prefix(suffix)-free, respectively.}\label{fig:reduction-relation}
\end{figure}
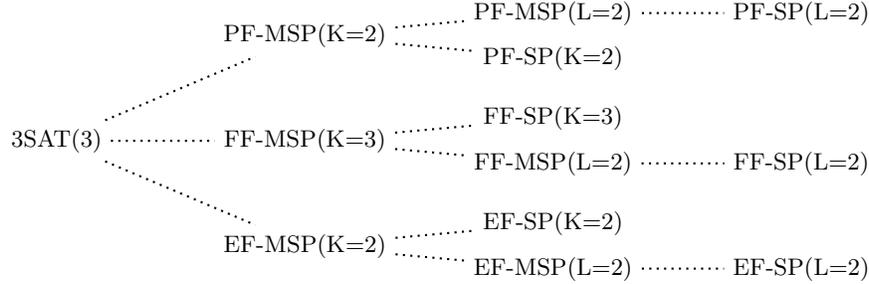

For each $\mathcal{X}$ in
$\{\text{equality},\text{prefix},\text{suffix},\text{factor}\}$, we
will consider two string partition problems. 

\begin{problem}[$\mathcal{X}$-Free Multiple String Partition
  ($\mathcal{X}$-MSP) Problem]\label{p:MSP}
  \ \\
  \noindent
  \emph{Instance}: Finite alphabet $\Sigma$ of size $L$, a positive
  integer $K$, and a set of strings $\mathcal{W}$ over $\Sigma^*$.\\
  \noindent
  \emph{Question}: Is there an $\mathcal{X}$-free, $K$-partition $P$ of
  $\mathcal{W}$?
\end{problem}

\begin{problem}[$\mathcal{X}$-Free String Partition
  ($\mathcal{X}$-SP) Problem]\label{p:SP}
  \ \\
  \noindent
  \emph{Instance}: Finite alphabet $\Sigma$ of size $L$, a positive
  integer $K$, and a string $w$ over $\Sigma^*$.\\
  \noindent
  \emph{Question}: Is there an $\mathcal{X}$-free, $K$-partition $P$ of
  $w$?
\end{problem}

We will show \NP-completeness of all these problems even when
restricted to the constant size of the partition ($K = 2,3$), or to
the binary alphabet ($L = 2$).  See
Figure~\ref{fig:reduction-relation} showing the chain of reductions
used to prove the complexity of the three variations and related
restrictions of the problem.

\section{Equality-Free String Partition Problems}
\label{sec:equal-free-part}

\subsection{Equality-Free Multiple String Partition with Unbounded Alphabet}
\label{sec:equal-free-mult-unbounded-alphabet}

We now describe a polynomial reduction from 3SAT(3) to EF-MSP with $K
= 2$ and unbounded alphabet.  Let
$\phi$ be an instance of 3SAT(3), with set $C=\{c_{1},\dots,c_{m}\}$
of clauses, and set $X=x_{1},\dots,x_{n}$ of variables.  We shall
define an alphabet $\Sigma$ and construct a set of strings
$\mathcal{W}$ over $\Sigma^*$,
such that $\mathcal{W}$ has a collision-free 2-partition if and only if $\phi$
is satisfiable. Let $|c_{i}|$ denote the number of literals contained
in the clause $c_{i}$
and let $c_i^1,\dots,c_i^{|c_i|}$ be the literals of clause $c_i$.

We construct $\mathcal{W}$ to be a union of three types of strings:
clause strings ($\mathcal{C}$), enforcer strings ($\mathcal{E}$) and
forbidden strings ($\mathcal{F}$). First, for each clause of $\phi$,
we create a clause string $C$ such that an equality-free 2-partition
of $\mathcal{C}$ unambiguously selects exactly one literal from $C$.
We refer to the selected strings corresponding to literals as
\emph{selected literals}.
Intuitively, the selected literals of the clause strings are intended
to be a satisfying truth assignment for the variables of $\phi$.
Second, for each variable we create an enforcer string to ensure that
selected literals are \emph{consistent}. Specifically, the enforcer
strings ensure that a positive and a negative literal for the same
variable cannot be simultaneously selected.  Finally, we find it
helpful to create so called forbidden strings that ensure certain strings
cannot be selected in the clause and enforcer strings.

We construct an alphabet $\Sigma$, formally defined below, which
includes a letter for each literal occurrence in the clauses, one
letter for each variable, and the letters $\boxminus $ and $\boxplus $
used as delimiters.
\begin{align*}
  \Sigma &= \{\var_i;\; x_i \in X\}
  \cup \{\lit_{i}^{j};\; c_i \in C \wedge 1 \leq j \leq |c_i|\}
  \cup \{\boxminus ,\boxplus \}
\end{align*}

Note that $|\Sigma|$ is linear in the size of the
3SAT(3) problem $\phi$ (at most $n+3m+2$).

\paragraph{Construction of forbidden strings:}

To ensure that certain strings cannot be selected in $\mathcal{C}$ or
$\mathcal{E}$, we will use the following set of forbidden strings
$\mathcal{F} = \{\boxminus ,\boxplus \}$.

\begin{observation}\label{obs:EF-MSP-forb}
  No string from the forbidden set $\mathcal{F}$
  can be selected in $\mathcal{C}$ or $\mathcal{E}$.
\end{observation}

\paragraph{Construction of clause strings:}  

For each clause $c_i \in C$, construct the \emph{$i$-th clause string}
to be $\lit_i^1 \boxminus \lit_i^2$ if $|c_i| = 2$, and
$\lit_i^1 \boxminus \lit_i^2 \boxminus \lit_i^3$ if $|c_i| = 3$.

\begin{figure}
  \centering \small
  \begin{tikzpicture}[xscale=0.67,scale=1.0]
    \begin{scope}

      \underscore[red]{1}{4}{0}
      \underscore{2}{4}{1}
      
      \underscore{1}{3.7}{1} 
      \underscore[red]{3}{3.7}{0}
       
      \draw (1,4) +(0,0) node{\ensuremath{\lit_{i}^1}};
      \draw (1,4) +(1,0) node{\ensuremath{\boxminus }};
      \draw (1,4) +(2,0) node{\ensuremath{\lit_{i}^2}};
    \end{scope}
    
    \begin{scope}[shift={(6,0)}]

      \underscore[red]{1}{4}{0}
      \underscore{2}{4}{1}
      \underscore{4}{4}{1}

      \underscore{1}{3.7}{1} 
      \underscore[red]{3}{3.7}{0}
      \underscore{4}{3.7}{1}

      \underscore{1}{3.4}{1} 
      \underscore{3}{3.4}{1}
      \underscore[red]{5}{3.4}{0}
      
      \draw (1,4) +(0,0) node{\ensuremath{\lit_{i}^1}};
      \draw (1,4) +(1,0) node{\ensuremath{\boxminus }};
      \draw (1,4) +(2,0) node{\ensuremath{\lit_{i}^2}};
      \draw (1,4) +(3,0) node{\ensuremath{\boxminus }};
      \draw (1,4) +(4,0) node{\ensuremath{\lit_{i}^3}};
    \end{scope}
    
  \end{tikzpicture}
  \caption{The 2-literal clause string (left) and 3-literal clause
    string (right) used in the reduction from 3SAT(3) to EF-MSP.
    Shown below each string are all valid 2-partitions.
    \emph{Selected} literals of a partition are shown in \textcolor{red}{red}.}
  \label{fig:EF-MSP-clause-string}
\end{figure}
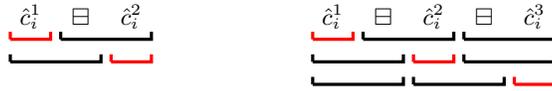

\begin{lemma}\label{lem:EF-MSP-clause-string} 
  Given that no string from the forbidden set $\mathcal{F}$ is
  selected in $\mathcal{C}$, exactly one literal
  letter must be selected for each clause string in any equality-free
  2-partition of $\mathcal{C}$.
\end{lemma}

\begin{proof}
  Consider the clause string for clause $c_i$.  
  Whether $c_i$ has two or three literals, 
  the forbidden substring $\boxminus $ cannot be selected alone.
  Therefore, each $\boxminus $ must be selected with an adjacent literal
  letter.  This leaves exactly one other literal letter which must be
  selected (see Figure \ref{fig:EF-MSP-clause-string}).\qed
\end{proof}

\paragraph{Construction of enforcer strings:}
 
We must now ensure that no literal of $\phi$ that is selected in $\mathcal{C}$
is the negation of another selected literal.  By definition of
3SAT(3), each variable appears exactly three times:
twice positive and once negated. Let $c_i^p$ and $c_j^q$ be
the two positive and $c_k^r$ the negated occurrences of a
variable $x_{v}$.  Then construct the \emph{enforcer string} for this
variable as follows
$\lit_i^p \boxplus \lit_k^r \var_v \lit_k^r \var_v \lit_k^r \boxplus \lit_j^q$.

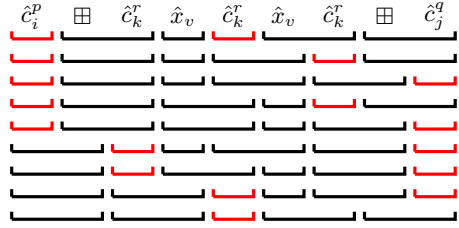
\begin{figure}
	\centering
  \small
  \begin{tikzpicture}[xscale=0.67,scale=1.0]
    \underscore[red]{1}{4.0}{0} \underscore{2}{4.0}{1}
    \underscore{4}{4.0}{0} \underscore[red]{5}{4.0}{0}
    \underscore{6}{4.0}{1} \underscore{8}{4.0}{1}

    \underscore[red]{1}{3.7}{0} \underscore{2}{3.7}{1}
    \underscore{4}{3.7}{0} \underscore{5}{3.7}{1}
    \underscore[red]{7}{3.7}{0} \underscore{8}{3.7}{1}

    \underscore[red]{1}{3.4}{0} \underscore{2}{3.4}{1} \underscore{4}{3.4}{0} 
    \underscore{5}{3.4}{1} \underscore{7}{3.4}{1} \underscore[red]{9}{3.4}{0}

    \underscore[red]{1}{3.1}{0} \underscore{2}{3.1}{1}
    \underscore{4}{3.1}{1} \underscore{6}{3.1}{0}
    \underscore[red]{7}{3.1}{0} \underscore{8}{3.1}{1}

    \underscore[red]{1}{2.8}{0} \underscore{2}{2.8}{1} \underscore{4}{2.8}{1} 
    \underscore{6}{2.8}{0} \underscore{7}{2.8}{1} \underscore[red]{9}{2.8}{0}

    \underscore{1}{2.5}{1} \underscore[red]{3}{2.5}{0} 
    \underscore{4}{2.5}{0} \underscore{5}{2.5}{1} 
    \underscore{7}{2.5}{1} \underscore[red]{9}{2.5}{0}

    \underscore{1}{2.2}{1} \underscore[red]{3}{2.2}{0} 
    \underscore{4}{2.2}{1} \underscore{6}{2.2}{0} 
    \underscore{7}{2.2}{1} \underscore[red]{9}{2.2}{0}

    \underscore{1}{1.9}{1} \underscore{3}{1.9}{1} 
    \underscore[red]{5}{1.9}{0} \underscore{6}{1.9}{0} 
    \underscore{7}{1.9}{1} \underscore[red]{9}{1.9}{0}

    \underscore{1}{1.6}{1} \underscore{3}{1.6}{1}
    \underscore[red]{5}{1.6}{0} \underscore{6}{1.6}{1}
    \underscore{8}{1.6}{1}

    \newcommand{\enfstr}[2] {
      \draw (#1,#2) +(0,0) node{\ensuremath{\lit_i^p}};
      \draw (#1,#2) +(1,0) node{\ensuremath{\boxplus }};
      \draw (#1,#2) +(2,0) node{\ensuremath{\lit_k^r}};
      \draw (#1,#2) +(3,0) node{\ensuremath{\var_v}};
      \draw (#1,#2) +(4,0) node{\ensuremath{\lit_k^r}};
      \draw (#1,#2) +(5,0) node{\ensuremath{\var_v}};
      \draw (#1,#2) +(6,0) node{\ensuremath{\lit_k^r}};
      \draw (#1,#2) +(7,0) node{\ensuremath{\boxplus }};
      \draw (#1,#2) +(8,0) node{\ensuremath{\lit_j^q}};
    }
    \enfstr{1}{4}
  \end{tikzpicture}
\caption{All possible 2-partitions are shown for the enforcer string 
	of a variable $x_v$ having two positive literals
  $c_i^p$ and $c_j^q$, and one negative literal
  $c_k^r$.  In each partition, either $\lit_{k}^{r}$ is selected or both
  $\lit_{i}^{p}$ and $\lit_{j}^{q}$ are which guarantees that
  letters for positive and negated literals of $x_{v}$ cannot be
  simultaneously selected in $\mathcal{C}$.}
  \label{fig:EF-MSP-enforcer-string}
\end{figure}

\begin{lemma}\label{lem:EF-MSP-enforcer-string}
  Given that no string from the forbidden set $\mathcal{F}$ is
  selected in $\mathcal{C}\cup \mathcal{E}$,
  any equality-free 2-partition of $\mathcal{C}\cup \mathcal{E}$ must
  be consistent. In addition, for any consistent choice of selecting
  letters for literals in $\mathcal{C}$, there is an equality-free
  2-partition of $\mathcal{C}\cup \mathcal{E}\cup \mathcal{F}$.
\end{lemma}

\begin{proof}
  Consider the enforcer string for variable $x_{v}$ with positive
  literals $c_{i}^{p}=c_{j}^{q}=x_{v}$, and the negated
  literal $c_k^r=\neg x_{v}$. Figure~\ref{fig:EF-MSP-enforcer-string} shows all 9
  possible 2-partitions of the enforcer string (since $\boxplus $ is a forbidden
  string, each $\boxplus $ must be selected with an adjacent letter). It follows that
  in each of them either $\lit_{k}^{r}$ is selected or both
  $\lit_{i}^{p}$ and $\lit_{j}^{q}$ are. In the first case,
  $\lit_{k}^{r}$ cannot be selected in $\mathcal{C}$ and thus satisfied
  literals are chosen consistently for $x_{v}$. In the second case,
  letters for neither of the positive occurrences of $x_{v}$ can be selected in
  $\mathcal{C}$. 

  To show the second part of the claim, observe that there is
  a 2-partition of the enforcer string compatible with any of four valid
  combinations of selecting letters for the corresponding literals in $\mathcal{C}$ (for
  example, by choosing the fifth or the last 2-partitions in
  Figure~\ref{fig:EF-MSP-enforcer-string}). Since enforcer strings
  share only one letter in common, namely, $\boxplus $, which is never
  selected in the enforcer strings, there are no collisions between
  2-partitions of all enforcer strings. Furthermore, there are no collisions
  between strings selected in $\mathcal{C}$ and in $\mathcal{E}$:
  strings of length two selected in $\mathcal{C}$ contain the letter
  $\boxminus $, which does not appear in the enforcer strings; strings
  of length one are literals and the partitioning of enforcer strings
  was chosen in a way that literals (in
  $\mathcal{C}$) cannot be selected again in $\mathcal{E}$.\qed
\end{proof}

This completes the reduction.  Notice that the reduction is polynomial
as the combined length of the constructed set of strings $\mathcal{W} =
\mathcal{C} \cup \mathcal{E} \cup \mathcal{F}$ is at most $5m + 9n + 2$.

\begin{theorem}\label{thm:EF-MSP-NPC}
  Equality-Free Multiple String Partition (EF-MSP) is
  \NP-complete for $K=2$.
\end{theorem}

\begin{proof}\label{pf:EF-MSP-NPC}
  It is easy to see that EF-MSP Problem is in \NP: a nondeterministic algorithm
  need only guess a partition $P$ where $|p_i| \leq K$ for all $p_i$
  in $P$ and check in polynomial time that no two strings in $P$ are
  equal. Furthermore, it is clear that an arbitrary
  instance $\phi$ of 3SAT(3) can be reduced to an instance of EF-MSP,
  specified by a set of strings $\mathcal{W}=\mathcal{C} \cup
  \mathcal{E} \cup \mathcal{F}$, in polynomial time and
  space by the reduction detailed above.

  Now suppose there is a satisfying truth assignment for $\phi$.
  Simply select one corresponding true literal per clause in
  $\mathcal{C}$. The construction of clause strings guarantees that a
  2-partition of the rest of each clause string is possible.  Also,
  since a satisfying truth assignment for $\phi$ cannot assign truth
  values to opposite literals, then Lemma \ref{lem:EF-MSP-enforcer-string}
  guarantees that a valid partition of the enforcer strings is
  possible which does not conflict with the clause strings.
  Therefore, there exists an equality-free multiple string
  partition of $\mathcal{W}$.

  Likewise, consider an equality-free multiple string partition of
  $\mathcal{W}$.  Lemma \ref{lem:EF-MSP-clause-string} ensures that at
  least one literal per clause is selected.  Furthermore, Lemma
  \ref{lem:EF-MSP-enforcer-string} guarantees that if there is no collision,
  then no two selected variables in the clauses are negations of each
  other.  Therefore, this must correspond to a satisfying truth
  assignment for $\phi$ (if none of the three literals of a variable
  is selected in the partition of $\mathcal{C}$ then this variable can
  have arbitrary value in the truth assignment without affecting
  satisfiability of $\phi$).\qed
\end{proof}

\subsection{Equality-Free String Partition with Unbounded Alphabet}
\label{sec:equal-free-unbounded-alphabet}

\begin{theorem}\label{thm:EF-SP-NPC}
  Equality-Free String Partition (EF-SP) is \NP-complete for
  $K = 2$.
\end{theorem}

\begin{proof}
  To show that EF-SP Problem for $K = 2$ is \NP-complete, we will reduce EF-MSP
  Problem for $K = 2$ to it. Consider an arbitrary instance $I$ of EF-MSP having
  a set of strings $\mathcal{W}=\{w_1,w_2,\ldots,w_\ell \}$ over
  alphabet $\Sigma$, and maximum partition size $K = 2$.  We construct
  an instance $\bar I$ of EF-SP as follows. Let $\widehat{\Sigma} =
  \{\boxdot \} \cup \{d_i\text{, for $1 \leq i < \ell $}\}$, where
  $\widehat{\Sigma} \cap \Sigma = \emptyset$.  Set the alphabet of
  $\bar I$ to $\bar \Sigma = \Sigma \cup \widehat{\Sigma }$ and the
  maximum partition size to $\bar K = 2$.  Note that $|\bar \Sigma| =
  |\Sigma| + \ell $.  Finally, construct the string
  \begin{equation*}
    \bar w = \boxdot \boxdot \boxdot \boxdot \boxminus 
    w_1 
    d_{1}\boxdot \boxdot d_{1}
    w_{2}
    d_{2}\boxdot \boxdot d_{2}
    \dots
    d_{\ell - 1 }\boxdot \boxdot d_{\ell - 1 }
    w_{\ell }\,.
  \end{equation*}
  The prefix of $\bar w$ of length five can be partitioned in two different
  ways each selecting $\boxdot $.
  Consequently, in any
  2-partition of $\bar w$, remaining occurrences of $\boxdot $ must be
  selected together with an adjacent letter different from $\boxdot $,
  i.e., all strings $d_{i}\boxdot $ and $\boxdot d_{i}$ must be
  selected. Therefore, any 2-partition of $\bar w$ contains a
  2-partition of $\mathcal{W}$ and the strings
  $\mathcal{D} = \{\boxdot,\boxdot\boxdot, \boxdot\boxminus,d_{1}\boxdot,\boxdot
  d_{1},\dots,d_{\ell - 1}\boxdot,\boxdot d_{\ell - 1}\}$. On the
  other hand, since all strings in $\mathcal{D}$ contain $\boxdot
  \notin \Sigma $, any 2-partition of $\bar w$ together with $\mathcal{D}$
  forms a 2-partition of $\mathcal{W}$. It follows that there is a
  2-partition of $\mathcal{W}$ if and only if there is a 2-partition
  of $\bar w$. 
  The reduction is in polynomial time and space as $|\bar w|=||\mathcal{W}|| + 4\ell + 1$. 
  \qed
\end{proof}

\subsection{Equality-Free Multiple String Partition with Binary Alphabet}
\label{sec:equal-free-mult-binary-alphabet}

\begin{theorem}\label{thm:EF-MSP-binary}
  The EF-MSP with maximum partition size $K = 2$ can be
  polynomially reduced to the EF-MSP Problem with the alphabet size $L =
  2$. Consequently, the EF-MSP is \NP-complete for binary
  alphabet. In addition, this reduction satisfies the following
  property: for any set $C$ containing $n$ distinct strings of length
  $\delta $, where $n$ is the size of the alphabet of the EF-MSP
  with maximum partition size $K = 2$ and
  $\delta \ge \log_{2} n$, every selected word in a valid
  partition (if it exists) of the EF-MSP with the binary
  alphabet is a prefix of a string in $C^{2}$, and its
  maximum partition size is $\bar K = 2\delta $.
\end{theorem}

\begin{proof}
  We will show a reduction from the EF-MSP with maximum
  partition size $K = 2$. Consider an arbitrary instance $I$ of EF-MSP
  having a set of strings $\mathcal{W}=\{w_1,w_2,\ldots,w_\ell \}$
  over alphabet $\Sigma = \{a_{1},\dots,a_{n}\} $, and maximum
  partition size $K = 2$. We will construct an instance $\bar I$ of
  EF-MSP over binary alphabet $\bar \Sigma = \{0,1\} $. Let $\delta $
  be any number greater or equal to $\log_{2} n$. Let $C =
  \{c_{1},\dots,c_{n}\} $ be a set of any distinct binary codewords of
  length $\delta $. We set $\bar K$ to $2\delta $. Let $h$ be a
  homomorphism from $\Sigma $ to $C$ such that $h(a_{i}) = c_{i}$, for
  every $i = 1,\dots,n$. The set of strings of $\bar I$ will contain
  $h(\mathcal{W})$, i.e., the original strings in $\mathcal{W}$ mapped
  by $h$ to the binary alphabet $\bar \Sigma $. However, we need to
  guarantee that the partition of strings in $h(\mathcal{W})$ does not
  contain fragments of codewords. For this reason, we also add to $\bar
  {\mathcal{W}}$ the following strings:
  \begin{align*}
    \widehat {\mathcal{W}} = \; &\{\pref_{i} (c);\; c\in C, i = 1,\dots,\delta - 1 \}\; \cup\\
    &\{\pref_{i} (cd);\; c,d\in C, i = \delta + 1 ,\dots,2\delta - 1 \}
  \end{align*}
  We set $\bar{\mathcal{W}} = h(\mathcal{W})\cup
  \widehat{\mathcal{W}}$. 

  First, consider a valid 2-partition $P$ of $\mathcal{W}$. We construct a
  $\bar K$-partition $\bar P$ of $\bar{\mathcal{W}}$ as follows. For
  each string $s$ selected in $P$, we select the corresponding $h(s)$
  in $\bar P$. For each string $t\in \widehat{\mathcal{W}}$, we select
  $t$ entirely. Note that strings selected from $h(\mathcal{W})$ have
  length either $\delta $ or $2\delta $, while strings selected from
  $\widehat{\mathcal{W}}$ have lengths different from $\delta $ and
  $2\delta $. Therefore, there cannot be any collisions between these
  two groups of selected strings. Furthermore, there are no collisions
  in the first group, since there were no collisions in
  $P$. Obviously, there are no collisions in the second group of selected
  strings. It follows that $\bar P$ is a valid $\bar K$-partition of
  $\bar{\mathcal{W}}$.

  Conversely, consider a valid $\bar K$-partition $\bar P$ of
  $\bar{\mathcal{W}}$. First, we will show that all strings in
  $\widehat{\mathcal{W}}$ are selected without non-trivial cut
  points. We will prove that by induction on the length $i$ of
  strings. The base case, $i = 1$, is trivially true, as one-letter
  strings cannot be partitioned into shorter strings. Now, assume the
  claim is true for all strings in $\widehat{\mathcal{W}}$ of lengths
  smaller than $i < 2\delta $ and different from $\delta $. Consider a
  word $u\in \widehat{\mathcal{W}}$ of length $i$. Assume that
  $u$ is partitioned into strings $u_{1},\dots,u_{t}$,
  where $t\ge 2$. Note that the length of $u_{1}$ is smaller than
  $i$. If the length of $u_{1}$ is different from $\delta $, we have a
  collision, as $u_{1}\in \widehat{\mathcal{W}}$ and by the induction
  hypothesis, it was selected without non-trivial cut points. Assume
  that the length of $u_{1}$ is $\delta $. Then $u_{2}$ is a prefix of
  a codeword of length smaller than $\min \{\delta ,i\}$, and we have a
  collision again as in the previous case. It follows that $t = 1$,
  i.e.,  $u$ is selected without non-trivial cut points in $\bar P$. Second,
  we show that all strings selected in the partition of strings in
  $h(\mathcal{W})$ have lengths either $\delta $ or $2\delta $. Assume
  that this is not the case for some string $s\in
  h(\mathcal{W})$. Note that $s = c_{i_{1}}c_{i_{2}}\dots c_{i_{p}}$,
  for some indices $i_{1},\dots,i_{p}$. Let $s = s_{1}\dots s_{q}$ be
  the partition of $s$ and let $j$ be the smallest $j$ such that the
  length of $s_{j}$ is not $\delta $ or $2\delta $. Then $s_{1}\dots
  s_{j - 1} = c_{i_{1}}\dots c_{i_{r}}$, for some $r < p$. Consequently, $s_{j}$ is a
  prefix of $c_{i_{r + 1}}c_{i_{r + 2}}$, i.e., $s_{j}\in
  \widehat{\mathcal{W}}$, and we have a collision, since $s_{j}$ was
  already selected in partition of $\widehat{\mathcal{W}}$. Hence,
  each string in $h(\mathcal{W})$ is partitioned into strings of
  lengths either $\delta $ or $2\delta $, which can be easily mapped
  to a valid 2-partition of $\mathcal{W}$.

  It follows that there is a 2-partition of $\mathcal{W}$ if and only
  if there is $\bar K$-partition of $\bar{\mathcal{W}}$ and that the
  reduction satisfies the property described in the claim. 

  Finally, let
  us check that the reduction is polynomial. The size of
  $h(\mathcal{W})$ is $|\mathcal{W}|$ and the length of
  $h(\mathcal{W})$ is $\delta ||\mathcal{W}||$. The size of
  $\widehat{\mathcal{W}}$---the set of all unique prefixes for codewords of
  length less than $\delta$, and all unique prefixes of pairs of
  adjacent codewords with length greater than $\delta$ and less than
  $2\delta$---is at most $(n^{2} + n)(\delta - 1)$  as there are $n$
  codewords in total.  Therefore, the length of
  $\widehat{\mathcal{W}}$ is at most $n\cdot (1 + \dots + \delta - 1)
  + n^{2}\cdot (\delta + 1 + \delta + 2 + \dots + 2 \delta - 1 ) = (3n^{2} +
  n)(\delta - 1 )\delta/2 $. Since $\delta $ can be chosen to be
  $\Theta (\log n)$, the size of $\bar{\mathcal{W}}$ is polynomial
  in the size of $\mathcal{W}$ and the size of the original alphabet
  $\Sigma $.
  \qed
\end{proof}

\subsection{Equality-Free String Partition with Binary Alphabet}
\label{sec:equal-free-binary-alphabet}

\begin{theorem}\label{thm:EF-SP-binary}
  Equality-Free String Partition (EF-SP) Problem is
  \NP-complete for binary alphabet ($L = 2$).
\end{theorem}

\begin{proof}
  We will show a reduction from the EF-MSP Problem with the binary
  alphabet ($L = 2$) satisfying properties listed in Theorem~\ref{thm:EF-MSP-binary}. 
  Consider an instance $I$ of EF-MSP
  having a set of strings $\mathcal{W}=\{w_1,w_2,\ldots,w_\ell \}$
  over alphabet $\Sigma = \{0,1\} $, and maximum partition size $K =
  2\delta $ such that all selected words in any valid $K$-partition
  are prefixes of the elements of a set $C^{2}$, where $C$ contains $n$ distinct
  strings of length $\delta $ each starting with $0$, $\ell \le (n^{2}
  + n)(\delta - 1)$, and $\delta \ge \max (9,3 \log_{2} (n + 1))$. By
  Theorem~\ref{thm:EF-MSP-binary}, this instance can be polynomially reduced to an
  instance of the EF-MSP with maximum partition size $K = 2$. We will
  construct an instance $\bar I$ of EF-SP over binary alphabet $\bar
  \Sigma = \{0,1\} $ with the same partition size $K = 2\delta $. We
  will show that the size of $\bar I$ is polynomial in the size of
  $I$, and hence, it will follow by Theorems~\ref{thm:EF-MSP-NPC}
  and~\ref{thm:EF-MSP-binary}, that the EF-SP Problem is \NP-complete.

  To construct the string $\bar w$ we will interleave strings
  $w_{1},\dots,w_{\ell }$ with delimiters $d_{1},\dots,d_{\ell - 1}$
  defined in a moment as follows:
  \begin{equation*}
    \bar w = w_{1}d_{1}w_{2}d_{2}w_{3}\dots d_{\ell - 1 }w_{\ell }\,.
  \end{equation*}
  To define the delimiter strings, we will need the following
  functions. Let $\bin :\: \mathbb{N}\to \{0,1\}^{*} $ be a function
  mapping a positive integer to its standard binary representation without the
  leading one. For example $\bin (1) = \varepsilon $, $\bin (2) = 0$
  and $\bin (10) = 010$. Next, the functions $\pad_{i} :\:
  \{0,1\}^{*}\to \{0,1\}^{*} $ will pad a given string with $i - 1$
  ones and one zero on the left, i.e., $\pad_{i} (s) = (1)^{i -
    1}0s$. We will refer to strings returned by this functions as
  \emph{padded strings}. The function $\chain :\: \{0,1\}^{*}\to
  \{0,1\}^{*} $ maps a string $s$ with $i$ trailing zeros, i.e., $s =
  s'(0)^{i}$, where $s'$ is either the empty string or a string ending
  with $1$, to the following concatenation of padded strings and
  mirror images (reversals) of padded strings:
  \begin{gather*}
    \chain (s) = \pad_{K - |s|}^{\mirror } (s)
    \pad_{K - |s|} (s')\pad_{K - |s|}^{\mirror } (s')
    \pad_{K - |s|} (s'0)\pad_{K - |s|}^{\mirror } (s'0)\\
    \dots
    \pad_{K - |s|} (s'(0)^{i - 1})\pad_{K - |s|}^{\mirror } (s'(0)^{i - 1})
    \pad_{K - |s|} (s) \,.
  \end{gather*}
  Finally, we set the delimiter $d_{j}$ to $\chain (\bin (j))$, for
  every $j > 1$. For $j = 1$, we set $d_{1}$ to $0(1)^{K - 1}(1)^{K(K
    - 1)/2}(1)^{K - 1}0$. To illustrate this definition, let us list the first five delimiter
  strings:
  \begin{align*}
    d_{1}& = 0(1)^{K - 1}(1)^{K(K - 1)/2}(1)^{K - 1}0 \\
    d_{2}& = \chain (0) = 00(1)^{K - 2}(1)^{K - 2}00(1)^{K - 2}(1)^{K - 2}00 \\
    d_{3}& = \chain (1) = 10(1)^{K - 2}(1)^{K - 2}01 \\
    d_{4}& = \chain (00) = 000(1)^{K - 3}(1)^{K - 3}00(1)^{K - 3}(1)^{K - 3}0000(1)^{K - 3}(1)^{K - 3}000 \\
    d_{5}& = \chain (01) = 100(1)^{K - 3}(1)^{K - 3}001
  \end{align*}

  Now, consider a valid $K$-partition $P$ of $\mathcal{W}$. We
  construct a $K$-partition $\bar P$ of $\bar{w}$ as follows. Each
  substring $w_{j}$ is partitioned in the same way as in $P$. Each
  delimiter $d_{j}$, where $j > 1$, is partitioned to its padded
  strings and mirror images of padded strings. In addition, the
  delimiter $d_{1}$ is partitioned into one mirror image of a padded
  string, strings $(1),(1)^{2},\dots,(1)^{K}$ in any order, and one
  padded string. Note that all strings selected in $w_{j}$'s are
  prefixes of $C^{2}$, and since each $c\in C$ has length $\delta =
  K/2$ and starts with $0$, all these selected strings start with $0$
  and the longest run of $1$ they contain has length at most $\delta -
  1$. Hence, they cannot collide with strings
  $(1),(1)^{2},\dots,(1)^{K}$ and with padded strings which all start with
  $1$. To show they do not collide with mirror images of padded
  strings, we will show that each padded string (or its mirror image)
  contains a run of at least $\delta $ ones. By the definition of
  functions $\pad_{i} $, each padded string or its mirror image
  selected in a delimiter $d_{j}$ contains a substring $(1)^{K - |\bin
    (d_{j})| - 1}$, i.e., a run of $K-(\lceil \log_{2} j \rceil - 1)
  -1 = K - \lceil \log_{2} j\rceil $
  ones. Since $j < \ell \le (n^{2} + n)(\delta - 1 )$, it is enough to
  show that $\log_{2} [(n^{2} + n)(\delta -1)]\le \delta $. This
  follows from the fact that $\delta \ge 2\log_{2} (n + 1) + \delta/3
  $ and $\delta/3 \ge \log_{2} (\delta - 1 )$ for $\delta \ge
  9$. Finally, we need to show that all selected padded strings and
  their mirror images are distinct. Note that each selected padded
  string starts with at least $\delta $ ones and contains at least one
  zero, hence, it cannot be equal to a selected mirror image of padded
  string. Hence, it is enough to show that two delimiter $d_{j}$ and
  $d_{j'}$, where $j,j' < \ell $ do not contain the same padded string
  or its mirror image. Without loss of generality, let us only
  consider the padded strings. If $\bin (j)$ and $\bin (j' )$ have
  different lengths then the padded strings of $d_{j}$ and $d_{j'}$
  start with $(1)^{K - |\bin (j)| - 1}0$ and $(1)^{K - |\bin (j')| -
    1}0$, hence they cannot be equal. Therefore, assume they have the
  same length. Let $s$ (respectively, $s'$) be the prefix of $\bin
  (j)$ (respectively, $\bin (j')$) without the trailing
  zeros. Clearly, $s\ne s'$. Now, the padded strings from $d_{j}$ and
  $d_{j'}$ are same only if $s0^{i} = s'0^{i'}$ for some $i$ and
  $i'$. However, since both $s$ and $s'$ end with one or one of them
  is the empty string, we must have $i
  = i'$, and hence also $s = s'$, a contradiction. Since the
  $K$-partition $P$ of $\mathcal{W}$ was valid, it follows that the
  $K$-partition of $\bar w$ is also valid.

  Conversely, consider a valid $K$-partition $\bar P$ of
  $\bar{w}$. It is enough to show that $\bar P$ super-selects each
  delimiter in $\bar w$. We will show by induction on $j$ that
  delimiters $d_{1},\dots,d_{j}$ are super-selected and furthermore,
  that each of these delimiters is partitioned into its padded strings
  and mirror images of padded strings. For the base case $j = 1$, it
  is easy to see that $\bar P$ must select string $0(1)^{K - 1}$, then
  strings $(1)^{1},\dots,(1)^{K}$ in any order and string $(1)^{K -
    1}0$, and thus $d_{1}$ is super-selected in $\bar P$ and its
  padded string and its mirror image of a padded string are selected. Next, assume
  that the induction hypothesis is satisfied for delimiters
  $d_{1},\dots,d_{j - 1}$. Consider delimiter string
  $d_{j}$. First, we will show that $d_{j}$ contains cut points in
  $\bar P$ shown by $\cdot $'s below:
  \begin{gather*}
    \pad_{K - |s|}^{\mirror } (s)\cdot 
    \pad_{K - |s|} (s')\pad_{K - |s|}^{\mirror } (s')\cdot 
    \pad_{K - |s|} (s'0)\pad_{K - |s|}^{\mirror } (s'0)\cdot \\
    \ldots
    \cdot \pad_{K - |s|} (s'(0)^{i - 1})\pad_{K - |s|}^{\mirror } (s'(0)^{i - 1})\cdot 
    \pad_{K - |s|} (s) \,,
  \end{gather*}
  where $s = \bin (j)$ and $s'$ is the prefix of $s$ without the
  trailing zeros and $i$ is the number of trailing zeros. Note that
  each letter ``$\cdot $'' is preceded and followed by $K - |\bin (j)| - 1$
  ones. Since $|\bin (j)|\le \delta - 1$, we have a run of at least $K
  = 2\delta $ ones, thus this run must contain a cut point. By
  contradiction assume that there is a cut point before the letter
  ``$\cdot $'' in this run of ones. Then the selected string starting at
  this cut point is in the form $(1)^{K - i - 1}0u$, where $i < |\bin
  (j)|$ and $|u|\le i$. Note that $u$ might be the empty string and
  the selected string must contain the zero preceding $u$ since all
  strings consisting only of ones are already selected in $d_{1}$. Let
  $v = u(0)^{i - |u|}$. Since $|v| = i < |\bin (j)|$, we have $v =
  \bin (j')$, where $j' < j$. The delimiter string $d_{j'}$ contains
  $\pad_{K - i} (u) = (1)^{K - i - 1}0u$, which by the induction
  hypothesis has been already selected. Analogously, we arrive into a
  contradiction, if there is a cut point after ``$\cdot $'' in the run of
  ones surrounding the letter ``$\cdot $''. It follows that there is a cut
  point at each letter ``$\cdot $'' above in $\bar P$. 

  Next, we show that each of super-selected strings of $d_{j}$:
  \begin{equation*}
    \pad_{K - |s|} (s')\pad_{K - |s|}^{\mirror } (s'), \dots, \pad_{K -
      |s|} (s'(0)^{i - 1})\pad_{K - |s|}^{\mirror } (s'(0)^{i - 1})\,,
  \end{equation*}
  has a cut point exactly in the middle. The length of each padded
  string or of its mirror image is at least $K - |s|$ and since $|\bin
  (j)|\le \delta - 1$, this length is at least $\delta + 1$. Hence,
  there has to be at least one cut point in each of the above
  super-selected strings in $\bar P$. We will first prove the claim
  for the first super-selected string $\pad_{K - |s|} (s')\pad_{K -
    |s|}^{\mirror } (s')$. By contradiction, and without loss of
  generality, assume that there is a cut point inside $\pad_{K - |s|}
  (s') = (1)^{K - |s| - 1}0s'$. Thus a string in the form $(1)^{K -
    |s| - 1}0u$, where $u$ is a proper prefix of $s'$, is selected in
  $\bar P$. Consider string $v = u(0)^{|s| - |u|}$. Obviously, $|v| =
  |s|$ and $v$ is
  lexicographically smaller than $s$, and thus $\bin (j') = v$ for
  some $j' < j$. By the induction hypothesis, string $\pad_{K - |v|}
  (u) = (1)^{K - |s| - 1}0u$ has been already selected in $d_{j'}$, a
  contradiction. It follows by straightforward induction on $i$ that
  the remaining super-selected strings are partitioned exactly in the
  middle. Finally, observe that if there is a cut point inside
  $\pad_{K - |s|} (s)$ then either one the padded strings of $d_{j'}$
  or one of the padded strings of $d_{j}$ described
  above is selected again .  Similarly, there cannot be any cut point inside $\pad_{K -
    |s|}^{\mirror } (s)$. Since the length of these two strings is
  exactly $K$, there has to be a cut point just after $\pad_{K - |s|}
  (s)$ and just before $\pad_{K - |s|}^{\mirror } (s)$, i.e., $d_{j}$
  is super-selected. This completes the induction proof, and we have
  that all delimiter strings in $ \bar w$ are super-selected by $\bar
  P$, and thus $\bar P$ gives us also a partition of the set $\mathcal{W}$. 

  It follows that there is a $K$-partition of $\mathcal{W}$ if and
  only if there is $K$-partition of $\bar{w}$. Finally, let us check
  that the reduction is polynomial. The length of each padded string
  or its mirror image is at most $K$. The length of $d_{1}$ is $K(K +
  3)/2 < K^{2}$. String $\bin (j)$ for $1 < j < \ell $ has length at most
  $\delta - 1$, and hence each $d_{j}$ contains at most $2\delta = K$
  padded strings and mirror images of padded strings. Hence,
  $|d_{j}|\le K^{2}$. Thus, the total length of $\bar w$ is at most
  $||\mathcal{W}|| + \ell K^{2}$.
  \qed
\end{proof}

\section{Factor-, Prefix- and Suffix-Free String Partition Problems}

Here, we summarize the results for these partition problems.  Their
proof can be found in the appendix of this paper.

\begin{theorem}\label{thm:FF}
  Both Factor-Free Multiple String Partition (FF-MSP) and Factor-Free
  String Partition (FF-SP) are \NP-complete in the following two
  cases: (a) when the maximum partition size is $3$; and (b) when the
  alphabet is binary.
\end{theorem}

\begin{theorem}\label{thm:PF}
  Both Prefix(Suffix)-Free Multiple String Partition (PF-MSP) and Prefix(Suffix)-Free
  String Partition (PF-SP) are \NP-complete in the following two
  cases: (a) when the maximum partition size is $2$; and (b) when the
  alphabet is binary.
\end{theorem}

\long \def\FF {
\section{Factor-Free String Partition Problems}
\label{sec:factor-free-part}

\subsection{Factor-Free Multiple String Partition with Unbounded Alphabet}
\label{sec:factor-free-mult-unbounded-alphabet}

Let $\phi$ be an instance of 3SAT(3), with set
$C=\{c_{1},\dots,c_{m}\}$ of clauses, and set $X=x_{1},\dots,x_{n}$ of
variables.  We shall define an alphabet $\Sigma$ and construct a set
of strings $\mathcal{W}$ over $\Sigma^*$, such that $\mathcal{W}$ has
a factor-free 3-partition if and only if $\phi$ is satisfiable.
Let $|c_{i}|$ denote the number of literals contained in the clause
$c_{i}$ and let $c_{i}^{1},\dots,c_{i}^{|c_i|}$ be the literals of clause
$c_i$.

We construct $\mathcal{W}$ to be the union of three sets of strings:
clause strings ($\mathcal{C}$), enforcer strings ($\mathcal{E}$) and
forbidden strings ($\mathcal{F}$) with the same function as in the
equality-free case.
We construct an alphabet $\Sigma$, formally defined below, which
includes a letter for each literal
occurrence in the clauses, three letters for each variable,
 and the letters $0$ and $1$.
\begin{align*}
  \Sigma &= \{\hat{x}_i^j;\; x_i \in X \wedge 1 \leq j \leq 3\}
  \cup \{\lit_{i}^{j};\; c_i \in C \wedge 1 \leq j \leq |c_i|\}
  \cup \{0,1\}
\end{align*}

Note that $|\Sigma|$ is linear in the size of
the 3SAT(3) problem $\phi$ (at most $3m + 3n + 2$).

\begin{observation}
  Since every letter appears at least twice in $\mathcal{W}$, no
  selected string can be a single letter.
\end{observation}

\begin{figure}
  \centering \small
  \vspace*{-5mm}
  \begin{tikzpicture}[xscale=0.67,scale=1.0]

    \draw (1,4) +(0,-.5) node{\ensuremath{\hat x_{v}^3}};
    \draw (1,4) +(1,-.5) node{\ensuremath{0}};
    \draw (1,4) +(2,-.5) node{\ensuremath{\hat x_{v}^3}};
    \draw (1,4) +(7.5,-.5) node{$1\le v\le n$};
    \underscore{1}{3.5}{2}
  \end{tikzpicture}
  \caption{The set of forbidden strings, $\mathcal{F}$, used in the
    reduction from 3SAT(3) to FF-MSP.}
  \label{fig:FF-MSP-construct-F}
\end{figure}

\paragraph{Construction of forbidden strings:}

To ensure that certain strings cannot be selected in
$\mathcal{C}$ or $\mathcal{E}$, we construct a set of forbidden strings
$\mathcal{F}$ as shown in Figure~\ref{fig:FF-MSP-construct-F}.
Specifically, we forbid, for every variable
$v$, any factor of the string $\hat{x}_{v}^3 0 \hat{x}_{v}^3$.
The number of strings in $\mathcal{F}$ is $n$.

\begin{lemma}\label{lem:FF-MSP-forb}
  No string or factor of a string from the forbidden set $\mathcal{F}$
  can be selected in $\mathcal{C}$ or $\mathcal{E}$.
\end{lemma}

\begin{proof}
  Consider any string $f \in \mathcal{F}$.
  If $f$ is split into two or three selected strings, a single letter
  is selected, which is not possible.
  Regardless of the construction of
  $\mathcal{C}$ and $\mathcal{E}$, it follows that in any valid
  partition, since $f$ is selected in $\mathcal{F}$, a factor of $f$
  cannot be selected in $\mathcal{C}$ nor in $\mathcal{E}$.
  \qed
\end{proof}

\paragraph{Construction of clause strings:}  
For each clause $c_i \in C$, construct the \emph{$i$-th clause string}
to be $\lit_{i}^1 \lit_{i}^1 0 \lit_{i}^2 \lit_{i}^2$, if $|c_i|
= 2$ and $\lit_{i}^1 \lit_{i}^1 0 \lit_{i}^2 \lit_{i}^2 0
\lit_{i}^3 \lit_{i}^3$, if $|c_i| = 3$.

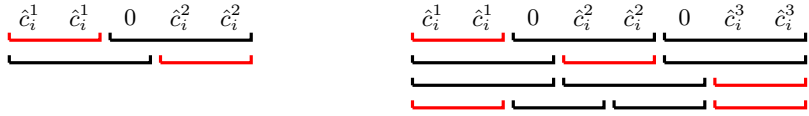
\begin{figure}
  \centering \small

  \begin{tikzpicture}[xscale=0.67,scale=1.0]
    \begin{scope}

      \underscore[red]{1}{4}{1}
      \underscore{3}{4}{2}
      
      \underscore{1}{3.7}{2} 
      \underscore[red]{4}{3.7}{1}
       
      \draw (1,4) +(0,0) node{\ensuremath{\lit_{i}^1}};
      \draw (1,4) +(1,0) node{\ensuremath{\lit_{i}^1}};
      \draw (1,4) +(2,0) node{\ensuremath{0}};
      \draw (1,4) +(3,0) node{\ensuremath{\lit_{i}^2}};
      \draw (1,4) +(4,0) node{\ensuremath{\lit_{i}^2}};
    \end{scope}
    
    \begin{scope}[shift={(8,0)}]

      \underscore[red]{1}{4}{1}
      \underscore{3}{4}{2}
      \underscore{6}{4}{2}

      \underscore{1}{3.7}{2} 
      \underscore[red]{4}{3.7}{1}
      \underscore{6}{3.7}{2}

      \underscore{1}{3.4}{2} 
      \underscore{4}{3.4}{2}
      \underscore[red]{7}{3.4}{1}

      \underscore[red]{1}{3.1}{1} 
      \underscore{3}{3.1}{1}
      \underscore{5}{3.1}{1}
      \underscore[red]{7}{3.1}{1}
      
      \draw (1,4) +(0,0) node{\ensuremath{\lit_{i}^1}};
      \draw (1,4) +(1,0) node{\ensuremath{\lit_{i}^1}};
      \draw (1,4) +(2,0) node{\ensuremath{0}};
      \draw (1,4) +(3,0) node{\ensuremath{\lit_{i}^2}};
      \draw (1,4) +(4,0) node{\ensuremath{\lit_{i}^2}};
      \draw (1,4) +(5,0) node{\ensuremath{0}};
      \draw (1,4) +(6,0) node{\ensuremath{\lit_{i}^3}};
      \draw (1,4) +(7,0) node{\ensuremath{\lit_{i}^3}};
    \end{scope}
    
  \end{tikzpicture}
  \caption{The 2-literal clause string (left) and 3-literal clause
    string (right) used in the reduction from 3SAT(3) to FF-MSP.
    Shown below each string are all valid partitions.  \emph{Selected}
    literals for a partition are shown
    in \textcolor{red}{red}.}
  \label{fig:FF-MSP-clause-string}
\end{figure}

\begin{lemma}\label{lem:FF-MSP-clause-string} 
  Given that no factor of a string from the forbidden set
  $\mathcal{F}$ is selected in $\mathcal{C} \cup \mathcal{E}$, at
  least one literal must be selected for each clause string in any
  factor-free 3-partition of $\mathcal{W}$.
\end{lemma}

\begin{proof}\label{pf:FF-MSP-clause-string}
  A literal letter cannot be selected alone without creating a
  collision.  Therefore, we say a literal $c_{i}^j$ is selected in
  clause $c_i$ if only if the string $\lit_{i}^j \lit_{i}^j$ is
  selected in the clause string for $c_{i}$.  Whether $c_i$ has two or
  three literals, a single letter $0$ cannot be selected.  A simple
  case analysis shows that in any valid partition at least one literal
  is selected (see Figure \ref{fig:FF-MSP-clause-string}).
  \qed
\end{proof}

\paragraph{Construction of enforcer strings:} 
We must now ensure that no literal of $\phi$ that is selected in
$\mathcal{C}$ is the negation of another selected literal.  By
definition of 3SAT(3), each variable appears exactly three times:
twice positive and once negated. Let $c_i^p$ and $c_j^q$ be the two
positive and $c_k^r$ the negated occurrences of variable $x_v$.
Then construct three \emph{enforcer strings} for this variable as
shown in Figure \ref{fig:FF-MSP-enforcer-str}.

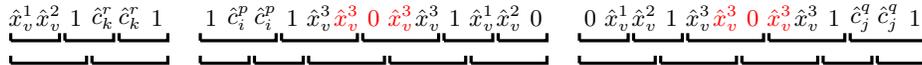
\begin{figure}
  \centering \small

  \begin{tikzpicture}[xscale=.36]

    \begin{scope}[shift={(0,0)}]

      \underscore{1}{4}{1}
      \underscore{3}{4}{1}
      \underscore{5}{4}{1}

      \underscore{1}{3.7}{2}
      \underscore{4}{3.7}{2}

      \draw (1,4) +(0,0) node{\ensuremath{\hat{x}_v^1}};
      \draw (1,4) +(1,0) node{\ensuremath{\hat{x}_v^2}};
      \draw (1,4) +(2,0) node{\ensuremath{1}};
      \draw (1,4) +(3,0) node{\ensuremath{\lit_{k}^r}};
      \draw (1,4) +(4,0) node{\ensuremath{\lit_{k}^r}};
      \draw (1,4) +(5,0) node{\ensuremath{1}};
    \end{scope}

    \begin{scope}[shift={(7,0)}]

      \underscore{1}{4}{2}
      \underscore{4}{4}{2}
      \underscore{7}{4}{2}
      \underscore{10}{4}{1}
      \underscore{12}{4}{1}

      \underscore{1}{3.7}{1}
      \underscore{3}{3.7}{1}
      \underscore{5}{3.7}{2}
      \underscore{8}{3.7}{2}
      \underscore{11}{3.7}{2}

      \draw (1,4) +(0,0) node{\ensuremath{1}};
      \draw (1,4) +(1,0) node{\ensuremath{\lit_{i}^p}};
      \draw (1,4) +(2,0) node{\ensuremath{\lit_{i}^p}};
      \draw (1,4) +(3,0) node{\ensuremath{1}};
      \draw (1,4) +(4,0) node{\ensuremath{\hat{x}_{v}^3}};
      \draw (1,4) +(5,0) node{\textcolor{red}{\ensuremath{\hat{x}_{v}^3}}};
      \draw (1,4) +(6,0) node{\textcolor{red}{\ensuremath{0}}};
      \draw (1,4) +(7,0) node{\textcolor{red}{\ensuremath{\hat{x}_{v}^3}}};
      \draw (1,4) +(8,0) node{\ensuremath{\hat{x}_{v}^3}};
      \draw (1,4) +(9,0) node{\ensuremath{1}};
      \draw (1,4) +(10,0) node{\ensuremath{\hat{x}_{v}^1}};
      \draw (1,4) +(11,0) node{\ensuremath{\hat{x}_{v}^2}};
      \draw (1,4) +(12,0) node{\ensuremath{0}};
    \end{scope}

    \begin{scope}[shift={(21,0)}]

      \underscore{1}{4}{1}
      \underscore{3}{4}{1}
      \underscore{5}{4}{2}
      \underscore{8}{4}{2}
      \underscore{11}{4}{2}

      \underscore{1}{3.7}{2}
      \underscore{4}{3.7}{2}
      \underscore{7}{3.7}{2}
      \underscore{10}{3.7}{1}
      \underscore{12}{3.7}{1}

      \draw (1,4) +(0,0) node{\ensuremath{0}};
      \draw (1,4) +(1,0) node{\ensuremath{\hat{x}_{v}^1}};
      \draw (1,4) +(2,0) node{\ensuremath{\hat{x}_{v}^2}};
      \draw (1,4) +(3,0) node{\ensuremath{1}};
      \draw (1,4) +(4,0) node{\ensuremath{\hat{x}_{v}^3}};
      \draw (1,4) +(5,0) node{\textcolor{red}{\ensuremath{\hat{x}_{v}^3}}};
      \draw (1,4) +(6,0) node{\textcolor{red}{\ensuremath{0}}};
      \draw (1,4) +(7,0) node{\textcolor{red}{\ensuremath{\hat{x}_{v}^3}}};
      \draw (1,4) +(8,0) node{\ensuremath{\hat{x}_{v}^3}};
      \draw (1,4) +(9,0) node{\ensuremath{1}};
      \draw (1,4) +(10,0) node{\ensuremath{\lit_j^q}};
      \draw (1,4) +(11,0) node{\ensuremath{\lit_j^q}};
      \draw (1,4) +(12,0) node{\ensuremath{1}};
    \end{scope}

  \end{tikzpicture}
  \caption{The enforcer strings for the three literals of variable $x_v$
    used in the reduction from 3SAT(3) to FF-MSP.  The two
    positive literals are denoted as $\lit_{i}^p$ and $\lit_{j}^q$ and
    the negative literal as $\lit_{k}^r$.  If the negative literal is
    selected in $\mathcal{C}$, then the enforcer string ensures
    neither positive literal can also be selected in $\mathcal{C}$
    without creating a collision (top row).  Likewise, if either or
    both of the positive literals are selected in $\mathcal{C}$, then
    the negative literal cannot be selected without creating a
    collision (bottom row).} 
  \label{fig:FF-MSP-enforcer-str}
\end{figure}

\begin{lemma}\label{lem:FF-MSP-enforce}
  Given that no factor of a string from the forbidden set
  $\mathcal{F}$ is selected in $\mathcal{C} \cup \mathcal{E}$, any
  factor-free 3-partition of $\mathcal{W}$ must be consistent.
\end{lemma}

\begin{proof}\label{pf:FF-MSP-enforce}
  Consider the three enforcer strings for some variable $x_v$ with positive
  literals $c_{i}^p=c_{j}^q=x_{v}$, and the negated literal
  $c_{k}^r=\neg x_{v}$ shown in Figure~\ref{fig:FF-MSP-enforcer-str}.
  Note that the red strings in the middle of the last two enforcer strings are
  forbidden, and hence, no partition can have a cut point at the
  beginning or the end of the red string. Note also that the factor
  $x_{v}^{3}x_{v}^{3}$ cannot be selected, as then
  $x_{v}^{3}x_{v}^{3}0$ or $0x_{v}^{3}x_{v}^{3}$ has
  to be selected, which is obviously not possible.
  Suppose the negative literal is selected in $\mathcal{C}$.  Then the
  only partition which can be selected without creating a collision
  selects strings containing both $\lit_{i}^p \lit_{i}^p$ and
  $\lit_{j}^q \lit_{j}^q$ as factors, thus forbidding them from being
  selected in $\mathcal{C}$ (see Figure~\ref{fig:FF-MSP-enforcer-str}
  (top partition)).
  Likewise, suppose one or both of the positive literals is selected
  in $\mathcal{C}$.
  Then only one partition of the first enforcer string
  is possible and it selects one string containing $\lit_{k}^r
  \lit_{k}^r$ as a factor, thus forbidding $\lit_{k}^r \lit_{k}^r$
  from being selected in $\mathcal{C}$ (see
  Figure~\ref{fig:FF-MSP-enforcer-str} (bottom partition)).  
  Note that
  while these enforcer strings ensure literals selected in the
  clauses are consistent, it also ensures unwanted collisions do not
  occur since selected strings containing literals are prefixed or
  suffixed by a 1, a letter not used in the clause strings. Also note
  that the selected strings containing the variable letters do not collide.
  \qed
\end{proof}

This completes the reduction.  Notice that the reduction is polynomial
as the combined length of the constructed set of strings $\mathcal{W} =
\mathcal{C} \cup \mathcal{E} \cup \mathcal{F}$ is at most $8m +
32n + 3n = 8m + 35n$.

\begin{theorem}\label{thm:FF-MSP-NPC}
  Factor-Free Multiple String Partition (FF-MSP) is \NP-complete.
\end{theorem}

\begin{proof}\label{pf:FF-MSP-NPC}
  It is easy to see that FF-MSP is in \NP: a nondeterministic algorithm
  need only guess a partition $P$ where $|p_i| \leq K$ for all $p_i$
  in $P$ and check in polynomial time that no string in $P$ is a
  factor of another.  Furthermore, it is clear that an arbitrary
  instance $\phi$ of 3SAT(3) can be reduced to an instance of FF-MSP,
  specified by a set of strings $\mathcal{W}=\mathcal{C} \cup
  \mathcal{E} \cup \mathcal{F}$, in polynomial time and
  space by the reduction detailed above.

  Now suppose there is a satisfying truth assignment for $\phi$.
  Simply select one corresponding true literal per clause in
  $\mathcal{C}$. The construction of clause strings guarantees that a
  3-partition of the rest of each clause string is possible.  Also,
  since a satisfying truth assignment for $\phi$ cannot assign truth
  values to opposite literals, then Lemma \ref{lem:FF-MSP-enforce}
  guarantees that a valid partition of the enforcer strings are
  possible which does not conflict with the clause strings.
  Therefore, there exists a factor-free multiple string
  partition of $\mathcal{W}$.

  Likewise, consider a factor-free multiple string partition of
  $\mathcal{W}$.  Lemma \ref{lem:FF-MSP-clause-string} ensures that at
  least one literal per clause is selected.  Furthermore, Lemma
  \ref{lem:FF-MSP-enforce} guarantees that if there is no collision,
  then no two selected variables in the clauses are negations of each
  other.  Therefore, this must correspond to a satisfying truth
  assignment for $\phi$ (if none of the three literals of a variable
  is selected in the partition of $\mathcal{C}$ then this variable can
  have arbitrary value in the truth assignment without affecting
  satisfiability of $\phi$).
  \qed
\end{proof}

\subsection{Factor-Free String Partition with Unbounded Alphabet}
\label{sec:factor-free-unbounded-alphabet}

\begin{lemma}\label{lem:FF-SP-connector}
  A valid $K$-partition $P$ of a string $w$ having $\alpha (x)^{3K-2}
  \beta$ as a factor, where $\alpha$ and $\beta$ are single letters
  other than $x$,
  must select the strings $\alpha (x)^{K-1}$, $(x)^{K}$, and
  $(x)^{K-1} \beta$.
\end{lemma}

\begin{proof}
  Three strings are required to cover the factor $\delta=\alpha
  (x)^{3K-2} \beta$.  However, more than three strings cannot be
  selected to cover $\delta$, as otherwise at least two factors are
  selected consisting only of the letter $x$ and must therefore
  collide.  There is only one partition of $\delta$ covered by three
  factors which must select $\alpha (x)^{K-1}$, $(x)^{K}$, and
  $(x)^{K-1} \beta$.
\end{proof}

\begin{theorem}\label{thm:FF-SP-NPC}
  Factor-Free String Partition (FF-SP) is \NP-complete.
\end{theorem}

\begin{proof}\label{pf:FF-SP-NPC}
  Consider an arbitrary instance $I$ of FF-MSP having a set of strings
  $\mathcal{W}=\{w_1,w_2,\ldots,w_n\}$ over alphabet $\Sigma$, and
  maximum partition size $K$.  We construct an instance $I'$ of FF-SP
  as follows.  Let $\widehat{\Sigma} = \{\alpha\} \cup
  \{\gamma_i\text{, for $1 \leq i < n$}\}$, where $\widehat{\Sigma} \cap
  \Sigma = \emptyset$.  Set the alphabet of $I'$ to $\Sigma' = \Sigma
  \cup \widehat{\Sigma}$ and the maximum partition size to $K' = K$.
  Note that $|\Sigma'| = |\Sigma| + n$.
  Finally, construct the string $w' = w_1 \alpha (\gamma_1)^{3K-2}
  \alpha w_2 \alpha (\gamma_2)^{3K-2} \alpha \ldots \alpha (\gamma_{n
    - 1})^{3K-2} \alpha w_n$.  The reduction
  is in polynomial time and space as $|w'|=\sum_{w_i \in
    \mathcal{W}}|w_i|+3K(n-1)$. 
  By Lemma
  \ref{lem:FF-SP-connector}, the factors $\alpha (\gamma_i)^{3K-2}
  \alpha$ of string $w'$ must be partitioned as $\alpha
  (\gamma_i)^{K-1}$, $(\gamma_i)^{K}$, $(\gamma_i)^{K-1} \alpha$, for
  $1 \leq i < n$.  Since any string containing a letter $\gamma_i$, $1
  \leq i < n$, cannot be a factor of any string in $\mathcal{W}$ it
  follows immediately that $w'$ has a $K'$-partition if and only if
  $\mathcal{W}$ has a $K$-partition.
  \qed
\end{proof}

\subsection{Factor-Free Multiple String Partition with Binary Alphabet}
\label{sec:factor-free-mult-binary-alphabet}

In this section we are going to reduce the size of alphabet to 2. In
order to do that we will map all letters of the original unbounded
alphabet $\Sigma$ except $0$ and $1$ to distinct binary
strings of length $t$ ($t$ has to be large enough so that we have
enough of strings), called codewords. Letters $0$ and $1$
will remain mapped to $0$ and $1$, respectively.
Consequently, $K$ will be set to $2t+1$. We will use the same clause
strings and a simplified version of the enforcer strings found in the
unbounded case (just mapped to the
binary alphabet). We will use only two forbidden strings, $000$ and $010$, to force
valid $K$-partitions to cut the clause and enforcer strings just
before or after the $0$ or $1$ letter. At
the end, we will show that this does not introduce any new collisions
in the $K$-partition corresponding to a truth assignment of the
3SAT(3) instance.

\paragraph{Construction of codewords:} We will use the codewords of
the following type $0(1)^{i}0(1)^{t-3-i}0$, where
$i\in\{2,\dots,t-5\}$. To make sure we have enough codewords for all
literal and variable letters (at most $3m + 2n$), we have to choose $t\ge3m + 2n + 6$.

\paragraph{Construction of forbidden string:} We will use only two
forbidden strings $\mathcal{F}=\{000,010\}$. Obviously, the only factor-free
partition of $\mathcal{F}$ is without any non-trivial cut points. These
two forbidden strings force any string containing $uav$ as a factor,
where $u$ and $v$ are codewords and $a$ is a letter, to contain exactly one cut point
around the letter $a$ between $u$ and $v$ as formalized in the
following lemma.

\begin{lemma}\label{l:cut-around-0}
  Any valid $K$-partitioning of $\alpha uav \beta$ and
  $\mathcal{F}$, where $u,v$ are codewords, $a\in \{0,1\}$ and
  $\alpha,\beta$ are arbitrary binary strings, contains either cut
  point $|\alpha|+|u|$ or $|\alpha|+|u|+1$, but not both.
\end{lemma}

\begin{proof}
  Assume we have a valid partitioning of $\alpha uav \beta$
  without cut points at positions $|\alpha|+|u|$ and $|\alpha|+|u|+1$.
  Since $u$ ends with $0$ and $v$ starts with $0$, there is a selected
  string containing $0a0$ as a factor, which is a forbidden
  string, a contradiction.
  \qed
\end{proof}

\paragraph{Construction of clause strings:} We will use the same
clause strings as for the unbounded alphabet:
\begin{equation*}
\lit_{i}^{1}\lit_{i}^{1}0\lit_{i}^{2}\lit_{i}^{2} \quad \text{or}\quad 
\lit_{i}^{1}\lit_{i}^{1}0\lit_{i}^{2}\lit_{i}^{2}0\lit_{i}^{3}\lit_{i}^{3}\,,
\end{equation*}
where $\lit_{i}^{j}$ are distinct codewords described above. We say
that a literal $c_{i}^{j}$ is \emph{selected} if
$\lit_{i}^{j}\lit_{i}^{j}$ is super-selected in the $K$-partition.

\begin{lemma}\label{l:clause-gadget}
  In any factor-free $K$-partition of $\mathcal{C}\cup\mathcal{F}$, at
  least one of the literals in each clause is selected.
\end{lemma}

\begin{proof}
  By Lemma~\ref{l:cut-around-0}, there is exactly one cut point around
  each of the $0$'s between codewords in the clause string. This means
  that the $K$-partition of the clause string follows one of the
  patterns depicted in Figure~\ref{fig:FF-MSP-clause-string} with the
  exception that any of the selected strings depicted in the figure
  can be further partitioned, \textit{i.e.}, they are super-selected. The claim
  follows.
  \qed
\end{proof}

\paragraph{Construction of  enforcer strings:} We will use a slightly
simplified version of the
enforcer strings as those used for the unbounded alphabet:
\begin{equation*}
  \hat x_{v}^1\hat x_{v}^21\lit_{k}^{r}\lit_{k}^{r}1,\quad 
  1\lit_{i}^{p}\lit_{i}^{p}1\hat x_{v}^1\hat x_{v}^20 \quad \text{and}\quad 
  0\hat x_{v}^1\hat x_{v}^21\lit_{j}^{q}\lit_{j}^{q}1\,, 
\end{equation*}
where
$\lit_{i}^{p}$ and $\lit_{j}^{q}$ are codewords for the positive
literals of a variable $x_{v}$, $\lit_{k}^{r}$ is the codeword for
the negated literal of $x_{v}$ and $\hat x_{v}^1,\hat x_{v}^2$ are
codewords for the variable letters of $x_{v}$.  The difference from
the unbounded case is that the factor $\hat{x}_v^3 \hat{x}_v^3 0
\hat{x}_v^3 \hat{x}_v^3 1$ can be safely removed from the second and
third strings without changing the logic of the gadgets due to the
property described in Lemma~\ref{l:cut-around-0}.

\begin{observation}\label{o:no-prefix-suffix-super-selection}
  In any factor-free $K$-partition no super-selected string can be a
  prefix (suffix) of any other super-selected string.
\end{observation}

\begin{lemma}\label{l:enforcer-gadget}
  In any factor-free $K$-partition of
  $\mathcal{C}\cup\mathcal{E}\cup\mathcal{F}$, the selected literals
  are consistent.
\end{lemma}

\begin{proof}
  By Lemma~\ref{l:cut-around-0}, either $\hat x_{v}^1\hat x_{v}^2$ or
  $\lit_{k}^{r}\lit_{k}^{r}1$ is super-selected in the first enforcer
  string. Note that if $\hat x_{v}^1\hat x_{v}^2$ is super-selected then
  by Observation~\ref{o:no-prefix-suffix-super-selection},
  $\hat x_{v}^1\hat x_{v}^20$ and $0\hat x_{v}^1\hat x_{v}^2$ cannot be
  super-selected.  Hence, if either $\hat x_{v}^1\hat x_{v}^20$ or
  $0\hat x_{v}^1\hat x_{v}^2$ is super-selected then in the first enforcer
  string $\lit_{k}^{r}\lit_{k}^{r}1$ is super-selected, and hence
  literal $c_{k}^{r}$ cannot be selected.

  By Lemma~\ref{l:cut-around-0}, either $1\lit_{i}^{p}\lit_{i}^{p}$ or
  $\hat x_{v}^1\hat x_{v}^20$ is super-selected in the second enforcer
  string. In the first case, by
  Observation~\ref{o:no-prefix-suffix-super-selection}, literal
  $c_{i}^{p}$ cannot be selected; in the second case, by the above
  argument, literal $c_{k}^{r}$ cannot be selected. Similarly, the
  last enforcer string ensures that literals $c_{j}^{q}$ and
  $c_{k}^{r}$ cannot be selected at the same time.
  \qed
\end{proof}

\begin{theorem}\label{t:FF-MSP-2}
  Factor-Free Multiple String Partition Problem for the binary
  alphabet (FF-MSP(2)) is \NP-complete.
\end{theorem}

\begin{proof}
  It follows by Lemmas~\ref{l:clause-gadget}
  and~\ref{l:enforcer-gadget} that if there is a factor-free
  $K$-partition of
  $\mathcal{W}=\mathcal{C}\cup\mathcal{E}\cup\mathcal{F}$, then the
  selected literals produce a satisfying assignment for the 3SAT(3)
  instance $\phi$.

  Now, assume that there is a satisfying assignment for $\phi$.
  Select a literal in each clause which satisfies it and partition
  all clause and enforcer strings accordingly, ensuring literals
  selected in the clause strings are not selected in the enforcer
  strings.
  We will show that this $K$-partition $P$ is
  factor-free. Obviously, the forbidden strings $000$ and $010$ are not factors of
  any selected string. In the clause strings, the $K$-partition $P$ selects
  the following types of strings: $0aa$, $aa0$ and $aa$, where $a$ is a
  codeword. In the enforcer string it selects the following types of
  strings: $1aa$, $aa1$, $0ab$, $ab0$, $ab1$, $ab$, $0a$, $a0$, $1a$,
  $a1$, where $a,b$ are distinct codewords. It follows by the proof in
  the unbounded case that two strings where one is obviously a factor
  of the other, like $aa$ and $0aa$, are not selected at the same
  time. Hence, it is enough to show that no new collisions are
  introduced by mapping the original letters to the binary alphabet.

  Obviously, the strings of the same length are different as all
  codewords are different and the strings of different types (e.g.,
  $0ab$ and $ab0$) would differ in the first two or last two letters.
  String $ab$ (where we can have $a=b$) is not a factor of some
  $\sigma cd$ ($cd\sigma$), where $a\ne c$, $b\ne d$ are codewords and
  $\sigma\in\{0,1\}$, as the $00$ in the middle of $ab$ would have to
  exactly match $00$ in the middle of $\sigma cd$ ($cd\sigma$), which
  would imply $ab=cd$. 

  Finally, we will show that $\sigma a$
  (and similarly, for $a\sigma$) is not a factor of $cd$ and $\sigma' cd$
  and $cd\sigma'$, where $a,c,d$ are codewords and $\sigma,\sigma'$ are
  letters, unless $\sigma=\sigma'$ and $a=c$. First, assume that
  $\sigma=1$. Then $\sigma a$ contains two $101$ factors with only
  $1$'s between them. String $cd$ contains exactly two $101$ factors but there
  is a $00$ factor between their occurrences in $cd$, and the same is true
  for $\sigma'cd$ and $cd\sigma'$ if $\sigma'=0$. Hence, assume that
  $\sigma'=1=\sigma$. Now, $cd\sigma'$ contains a factor
  $10(1)^{+}01$, but only at the very end, while in $\sigma a$ this
  factor is followed by at least two letters. Hence, the only
  possibility is that $\sigma a$ is a factor of $\sigma cd$. However,
  pattern $10(1)^{+}01$ appears only at the beginning of $\sigma cd$,
  and hence, $\sigma a$ would have to be a prefix of $\sigma
  cd$ and then $a=c$.

  Second, assume that $\sigma=0$. Then $\sigma a$ starts with
  $00$. Similar case analysis would show that $\sigma a$ can only be a
  factor of $\sigma ad$, or of $ca$, $\sigma'ca$ or
  $ca\sigma'$. However, string $0a$ can be only selected from the
  second enforcer string $0xy1bb1$ and $x$ never appears in the second
  position of any selected string of the type $cd$, $\sigma' cd$,
  $cd\sigma'$.
  \qed
\end{proof}

\subsection{Factor-Free String Partition with Binary Alphabet}
\label{sec:factor-free-binary-alphabet}

We will first design a sequence of strings which have to be selected no
matter where they appear in the string we are partitioning.

\begin{lemma}\label{l:delimiters}
  Let $K\ge1$ and for any $i\le K$, let $d_{i}=(1)^{i}0(1)^{K-1-i}$.
  Then any factor-free $K$-partition of
  \begin{equation*}
    w=u_{1}d_{0}(1)^{K}d_{0}^{\mirror}u_{2}d_{1}d_{1}^{\mirror}\dots
    d_{N-2}d_{N-2}^{\mirror}u_{N},
  \end{equation*}  
  where $N\le K/2$ and
  $u_{1},\dots,u_{N}$ are arbitrary strings, selects the following
  strings $(1)^{K},d_{i},d_{i}^{\mirror}$, for every $i=0,\dots,N-2$.
\end{lemma}

\begin{proof}
  Let $P$ be a factor-free multiple $K$-partition of $w$.  We will
  show by induction on $i$ that $(1)^{K}$, $d_{0},\dots,d_{i}$ and
  $d_{0}^{\mirror},\dots,d_{i}^{\mirror}$ are selected. The base case
  $i=0$ follows by Lemma~\ref{lem:FF-SP-connector}.  For the inductive
  step, assume that $(1)^{K}$, $d_{0},\dots,d_{i-1}$ and
  $d_{0}^{\mirror},\dots,d_{i-1}^{\mirror}$ are selected, where
  $i\le N-2$. We will show that $d_{i}$ and $d_{i}^{\mirror}$ are also
  selected. The factor $d_{i}d_{i}^{\mirror}$ of $w$ has length $2K$,
  hence, there is at least one cut point inside it.  Let $j$ be the
  first such cut point. Obviously, $j\le K$. Assume that $j<K$. Let
  $w_{p}$ be the selected string starting at cut point $j$.  We will
  consider two cases:

  \emph{Case 1.}  $j\ge i+1$. Then $w_{p}$ is a prefix of $(1)^{K}$,
  a contradiction since $(1)^{K}$ is already selected.

  \emph{Case 2.} $j\le i$. Then $w_{p}$ is a prefix of $d_{i-j}$, a
  contradiction since $d_{i-j}$ is already selected.

  Hence, the first cut point inside the factor $d_{i}d_{i}^{\mirror}$
  of $w$ is at position $K$. By symmetrical argument, this is also the
  last such cut point. It follows that both $d_{i}$ and
  $d_{i}^{\mirror}$ are selected.
  \qed
\end{proof}

\begin{corollary}\label{c:delimiters}
  Let $K\ge1$ and for any $i\le K$, let $d_{i}=(1)^{i}0(1)^{K-1-i}$.
  Consider the string
  \begin{equation*}
    w=u_{1}d_{0}(1)^{K}d_{0}^{\mirror}u_{2}d_{1}d_{1}^{\mirror}\dots
    d_{N-2}d_{N-2}^{\mirror}u_{N},
  \end{equation*}  
  where $N\le K/2$ and $u_{1},\dots,u_{N}$ are arbitrary strings.  If
  the string $w$ has a factor-free $K$-partition then the sequence of
  strings $u_{1},\dots,u_{N}$ has a factor-free multiple
  $K$-partition. On the other hand, if the sequence
  $u_{1},\dots,u_{N}$ has a factor-free multiple $K$-partition such
  that each selected string contains at least two $0$'s then $w$ has a
  factor-free $K$-partition.
\end{corollary}

\begin{proof}
  The first implication follows immediately by
  Lemma~\ref{l:delimiters}. The second implication follows by the
  fact that each delimiter contains only one $0$; hence, none of the
  selected strings in $u_{1},\dots,u_{N}$ can be a factor of a delimiter.
  \qed
\end{proof}

As the immediate consequence of Theorem~\ref{t:FF-MSP-2} and
Corollary~\ref{c:delimiters}, we have the following.

\begin{theorem}\label{t:FF-SP-2}
  Factor-free String Partition Problem for the binary alphabet
  (FF-SP(2)) is \NP-complete.
\end{theorem}
}

\long \def \PSF {
\section{Prefix/Suffix-Free String Partition Problems}
\label{sec:prefix-free-part}

All proofs presented in this section are for the prefix-free string
partition problems. The results for the suffix-free string problems
follow by symmetry.

\subsection{Prefix/Suffix-Free Multiple String Partition with Unbounded Alphabet}
\label{sec:prefix-free-mult-unbounded-alphabet}

Similar to the equality-free and factor-free cases, we will show a polynomial reduction
from an arbitrary instance of 3SAT(3).

Let $\phi$ be an instance of 3SAT(3), with set
$C=\{c_{1},\dots,c_{m}\}$ of clauses, and set $X=x_{1},\dots,x_{n}$ of
variables.  We shall define an alphabet $\Sigma$ and construct a set
of strings $\mathcal{W}$ over $\Sigma^*$, such that $\mathcal{W}$ has
a prefix-free 2-partition if and only if $\phi$ is satisfiable.
Let $|c_{i}|$ denote the number of literals contained in the clause
$c_{i}$ and let $c_{i}^{1},\dots,c_{i}^{|c_i|}$ be the literals of clause
$c_i$.

We construct $\mathcal{W}$ to be the union of three sets of strings:
clause strings ($\mathcal{C}$), enforcer strings ($\mathcal{E}$) and
forbidden strings ($\mathcal{F}$) with the same function as in the
equality-free and factor-free cases. 
We construct an alphabet $\Sigma$, formally defined below, which
includes four letters for each variable, a letter for each literal
occurrence in the clauses and the letter $\dollar$.
\begin{align*}
  \Sigma &= \{\var_i^j;\; x_i \in X \wedge 1 \leq j \leq 4\} \cup 
  \{\lit_{i}^{j};\; c_i \in C \wedge 1 \leq j \leq |c_i|\} \cup \{\dollar\}
\end{align*}

Note that $|\Sigma|$ is linear in the size of
the 3SAT(3) problem $\phi$ (at most $4n + 3m + 1$).

\paragraph{Construction of forbidden strings:} 

The forbidden set, $\mathcal{F}$, consists of the single string $\dollar
\dollar$.  Without loss of generality, we refer to this as the
\emph{forbidden string}.

\begin{lemma}\label{lem:PF-MSP-forb}
  No factor of the forbidden string
  can be selected in $\mathcal{C}$ nor in $\mathcal{E}$.
\end{lemma}

\begin{proof}
  No proper factor of the forbidden string can be selected without
  creating a collision.  Therefore, the entire string must be
  selected.  Regardless of the construction of $\mathcal{C}$ and
  $\mathcal{E}$, it follows that in any valid partition, since
  $\dollar \dollar$ is selected in $\mathcal{F}$, a factor of it cannot
  be selected in $\mathcal{C}$ nor in $\mathcal{E}$.
  \qed
\end{proof}

\paragraph{Construction of clause strings:}  
For each clause $c_i \in C$, construct the \emph{$i$-th clause string}
to be $\lit_{i}^1 \dollar \lit_{i}^2$, if $|c_i|
= 2$ and $\lit_{i}^1 \dollar \lit_{i}^2 \dollar \lit_{i}^3$, if $|c_i|
= 3$. 

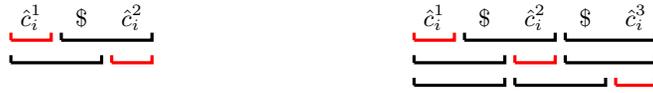
\begin{figure}
  \centering \small

  \begin{tikzpicture}[xscale=0.67,scale=1.0]
    \begin{scope}

      \underscore[red]{1}{4}{0}
      \underscore{2}{4}{1}
      
      \underscore{1}{3.7}{1} 
      \underscore[red]{3}{3.7}{0}
       
      \draw (1,4) +(0,0) node{\ensuremath{\lit_{i}^1}};
      \draw (1,4) +(1,0) node{\ensuremath{\dollar}};
      \draw (1,4) +(2,0) node{\ensuremath{\lit_{i}^2}};
    \end{scope}
    
    \begin{scope}[shift={(8,0)}]

      \underscore[red]{1}{4}{0}
      \underscore{2}{4}{1}
      \underscore{4}{4}{1}

      \underscore{1}{3.7}{1} 
      \underscore[red]{3}{3.7}{0}
      \underscore{4}{3.7}{1}

      \underscore{1}{3.4}{1} 
      \underscore{3}{3.4}{1}
      \underscore[red]{5}{3.4}{0}
      
      \draw (1,4) +(0,0) node{\ensuremath{\lit_{i}^1}};
      \draw (1,4) +(1,0) node{\ensuremath{\dollar}};
      \draw (1,4) +(2,0) node{\ensuremath{\lit_{i}^2}};
      \draw (1,4) +(3,0) node{\ensuremath{\dollar}};
      \draw (1,4) +(4,0) node{\ensuremath{\lit_{i}^3}};
    \end{scope}
    
  \end{tikzpicture}
  \caption{The 2-literal clause gadget (left) and 3-literal clause
    gadget (right) used in the reduction from 3SAT(3) to PF-MSP.
    Shown below each gadget are all valid partitions.  \emph{Selected}
    literals of a partition are shown
    in \textcolor{red}{red}.}
  \label{fig:PF-MSP-clause-string}
\end{figure}

\begin{lemma}\label{lem:PF-MSP-clause-string} 
  Given that no factor of the forbidden string
  is selected in $\mathcal{C} \cup \mathcal{E}$,
  exactly one literal must be selected for each clause string in any
  prefix-free 2-partition of $\mathcal{W}$.
\end{lemma}

\begin{proof}\label{pf:PF-MSP-clause-string}
  We say a literal $c_{i}^j$ is selected in clause $c_i$ if and only
  if the string $\lit_{i}^j$ is selected in the clause string for
  $c_{i}$.  Whether $c_i$ has two or three literals, the forbidden
  string $\dollar$ cannot be selected alone.  A simple case analysis
  shows that in any valid partition exactly one literal is selected
  (see Figure \ref{fig:PF-MSP-clause-string}).
  \qed
\end{proof}

\paragraph{Construction of enforcer strings:} 
We must now ensure that no literal of $\phi$ that is selected in
$\mathcal{C}$ is the negation of another selected literal.  By
definition of 3SAT(3), each variable appears exactly three times:
twice positive and once negated. Let $c_i^p$ and $c_j^q$ be the two
positive and $c_k^r$ the negated occurrences of a variable $x_{v}$.
Then construct two \emph{enforcer strings} for this variable as
shown in Figure \ref{fig:PF-MSP-enforcer-str}.

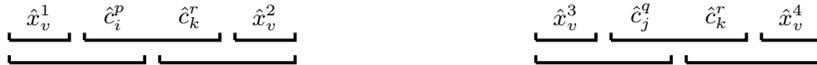
\begin{figure}
  \centering \small

  \begin{tikzpicture}[xscale=1,scale=1]

    \begin{scope}[shift={(0,0)}]

      \underscore{1}{4}{0}
      \underscore{2}{4}{1}
      \underscore{4}{4}{0}
      
      \underscore{1}{3.7}{1} 
      \underscore{3}{3.7}{1}
            
      \draw (1,4) +(0,0) node{\ensuremath{\var_v^1}};
      \draw (1,4) +(1,0) node{\ensuremath{\lit_{i}^p}};
      \draw (1,4) +(2,0) node{\ensuremath{\lit_{k}^r}};
      \draw (1,4) +(3,0) node{\ensuremath{\var_v^2}};
    \end{scope}

    \begin{scope}[shift={(7,0)}]
      \underscore{1}{4}{0} 
      \underscore{2}{4}{1}
      \underscore{4}{4}{0}

      \underscore{1}{3.7}{1}
      \underscore{3}{3.7}{1}
      
      \draw (1,4) +(0,0) node{\ensuremath{\var_v^3}};
      \draw (1,4) +(1,0) node{\ensuremath{\lit_{j}^q}};
      \draw (1,4) +(2,0) node{\ensuremath{\lit_{k}^r}};
      \draw (1,4) +(3,0) node{\ensuremath{\var_v^4}};
    \end{scope}
    
  \end{tikzpicture}
  \caption{The pair of enforcer strings for a variable $x_v$ used in
    the reduction from 3SAT(3) to PF-MSP.  The two positive literals
    for variable $x_v$ are denoted as $\lit_{i}^p$ and $\lit_{j}^q$
    and the negative literal as $\lit_{k}^r$.}
  \label{fig:PF-MSP-enforcer-str}
\end{figure}

\begin{lemma}\label{lem:PF-MSP-enforce}
  Given that no factor of the forbidden string
  is selected in $\mathcal{C} \cup \mathcal{E}$, any
  prefix-free 2-partition of $\mathcal{W}$ must be consistent.
\end{lemma}

\begin{proof}\label{pf:PF-MSP-enforce}
  Consider the two enforcer strings for variable $x_{v}$ with positive
  literals $c_{i}^p=c_{j}^q=x_{v}$, and the negated literal
  $c_{k}^r=\neg x_{v}$ shown in Figure~\ref{fig:PF-MSP-enforcer-str}.
 
  Suppose the negative literal is selected in $\mathcal{C}$.  Then the
  only partition which can be selected without creating a collision
  selects strings containing both $\lit_{i}^p$ and $\lit_{j}^q$ as a
  prefix, thus forbidding them from being selected in $\mathcal{C}$
  (see Figure~\ref{fig:PF-MSP-enforcer-str} (top row)).

  Likewise, if one or both of the positive literals is selected
  in $\mathcal{C}$ then in any collision-free 2-partition a string is
  selected containing $\lit_k^r$ as a prefix, thus forbidding the
  negative literal from being selected in $\mathcal{C}$ (see
  Figure~\ref{fig:PF-MSP-enforcer-str} (bottom row)).
  \qed
\end{proof}

This completes the reduction.  Notice that the reduction is linear as
the combined length of the constructed set of strings $\mathcal{W} =
\mathcal{C} \cup \mathcal{E} \cup \mathcal{F}$ is at most $5m + 8n +
2$.

\begin{theorem}\label{thm:PF-MSP-NPC}
  Prefix(Suffix)-Free Multiple String Partition (PF-MSP) is \NP-complete.
\end{theorem}

\begin{proof}\label{pf:PF-MSP-NPC}
  It is easy to see that PF-MSP is in \NP: a nondeterministic algorithm
  need only guess a partition $P$ where $|p_i| \leq K$ for all $p_i$
  in $P$ and check in polynomial time that no string in $P$ is a
  prefix of another.  Furthermore, it is clear that an arbitrary
  instance $\phi$ of 3SAT(3) can be reduced to an instance of PF-MSP,
  specified by a set of strings $\mathcal{W}=\mathcal{C} \cup
  \mathcal{E} \cup \mathcal{F}$, in polynomial time and
  space by the reduction detailed above.
	
  Now suppose there is a satisfying truth assignment for $\phi$.
  Simply select one corresponding true literal per clause in
  $\mathcal{C}$. The construction of clause strings guarantees that a
  2-partition of the rest of each clause string is possible.  Also,
  since a satisfying truth assignment for $\phi$ cannot assign truth
  values to opposite literals, then Lemma \ref{lem:PF-MSP-enforce}
  guarantees that a valid partition of the enforcer strings are
  possible.  Therefore, there exists a prefix-free multiple string
  partition of $\mathcal{W}$.
	
  Likewise, consider a prefix-free multiple string partition of
  $\mathcal{W}$.  Lemma \ref{lem:PF-MSP-clause-string} ensures that 
  exactly one literal per clause is selected.  Furthermore, Lemma
  \ref{lem:PF-MSP-enforce} guarantees that if there is no collision,
  then no two selected variables in the clauses are negations of each
  other.  Therefore, this must correspond to a satisfying truth
  assignment for $\phi$ (if none of the three literals of a variable
  is selected in the partition of $\mathcal{C}$ then this variable can
  have an arbitrary value in the truth assignment without affecting
  satisfiability of $\phi$).
  \qed
\end{proof}

\subsection{Prefix/Suffix-Free String Partition with Unbounded Alphabet}
\label{sec:prefix-free-unbounded-alphabet}

To show the single string restriction of this problem is \NP-complete,
we design the same delimiter strings as specified in the factor-free
construction.  The result follows immediately by Theorem
\ref{thm:PF-MSP-NPC} and Lemma \ref{lem:FF-SP-connector}.

\begin{theorem}\label{thm:PF-SP-NPC}
  Prefix(Suffix)-Free String Partition (PF-SP) is
  \NP-complete.
\end{theorem}

\subsection{Prefix/Suffix-Free Multiple String Partition with Binary Alphabet}
\label{sec:prefix-free-mult-binary-alphabet}

In this section we start with the same construction as the multiple string unbounded alphabet
case to form a set of strings $\mathcal{W}$, but show how the letters
can be encoded into binary.  We map the $\dollar$ letter to $1$ and
map all others letters of the original unbounded alphabet $\Sigma$ to
distinct binary strings of length $t$, called codewords.
Consequently, $K$ will be set to $2t$.  We will establish that no
codeword can properly contain a cut point.  Furthermore, by design, no
codeword is a prefix of another.  Since the mapping to binary does not
introduce new collisions, and since codewords cannot be cut in the
middle, the correctness of the construction will follow from the
results on the unbounded case.

\paragraph{Construction of codewords:} We use codewords of
the form $00(1)^{i}0(1)^{t-4-i}0$, where $i\in\{2,\dots,t-6\}$. To
ensure we have enough codewords for all literal and variable
letters (at most $3m + 4n$), we have to choose $t \ge 3m + 4n + 7$.

\paragraph{Construction of forbidden string:} We will use the
following set of forbidden strings: $\{11$, $01$, $101$, $0001$, $10001\}$.
Considering only the forbidden set, a simple case analysis
shows that each forbidden string must be entirely selected, otherwise
a collision occurs.  This set of forbidden strings ensures that no
codeword can be cut in the middle.

\begin{lemma}\label{lem:no-split-codewords}
  Given that no strings selected in
  $\mathcal{C} \cup \mathcal{E}$ have a forbidden prefix, any prefix-free $K$-partition of
  $\mathcal{W}$ must not contain a cut point within a codeword.
\end{lemma}

\begin{proof}
Recall that codewords are of length $t$ and that all length two binary
strings are prefixes of a forbidden string and therefore cannot be
selected.  Let us consider any cut point beginning within an arbitrary
codeword $w$.  For any proper suffix of $w$ longer than two, it
contains a prefix in the set $\{011, 111, 101, 110\}$.  Each of these
contains a forbidden string as a prefix and therefore a cut point in
$w$ cannot begin prior to positions $2,3,\ldots,t-2$.  We must now
show that a cut point cannot begin prior to position $t-1$ or prior to
position $t$.  Recall that by construction, $w$ is followed by either
another codeword, the letter $1$, or the empty string.  If the empty
string, it is not possible to a have a cut point in the position prior
to $t-1$ or $t$ since a string will be selected that has length less
than 3 and will therefore be a prefix of a forbidden string.  If $w$
is followed by the letter $1$ a cut point prior to $t-1$ or $t$ will
have a prefix in the set $\{101, 01\}$, both of which are forbidden
strings.  Finally consider the case that $w$ is followed by another
codeword and recall that all codewords begin with $001$.  If a
cut point occurs prior to position $t-1$, and the selected string
beginning at that position has length at least five, then it will
contain $10001$ as a prefix which is a forbidden string; any shorter
selection will be a prefix of the forbidden string $10001$.  If a
cut point occurs prior to position $t$ and the selected string has
length at least four, it will contain $0001$ as a prefix which is a
forbidden string; any shorter selection will be a prefix of the
forbidden string $0001$.\qed
\end{proof}

The above lemma ensures no codeword is divisible.  The result is that
the binary encoded instance $I'$ of an unbounded alphabet instance $I$
can be partitioned exactly in the same relative positions as the
original instance.  Since each codeword cannot be a prefix of another
by design, then correctness of the binary case immediately follows.

\begin{theorem}\label{thm:PF-MSP-NPC-binary}
  Prefix(Suffix)-Free Multiple String Partition (PF-MSP) is
  \NP-complete for binary alphabet ($L=2$).
\end{theorem}

\subsection{Prefix/Suffix-Free String Partition with Binary Alphabet}
\label{sec:prefix-free-binary-alphabet}

Similar to the factor-free case, we will design delimiters to join the
set $\mathcal{W}$ of the multiple string case, into one string,
without changing the possibilities for partitioning the original set
of strings and without introducing new types of collisions.
Specifically, we will create a new string instance $I=WF$ where $W$ is
a string that concatenates all strings in $\mathcal{W}$, expect from
the forbidden set $\mathcal{F}$, using delimiters we describe below.
The string $F$ has a special construction to ensure the strings from
the forbidden set must be selected in $F$ for any collision-free
partition of $I$.  Thus, the new instance $I$ will have a
collision-free partition if and only if $\mathcal{W}$ does.

\paragraph{Construction of delimiters:} 
We design delimiters similar to the codewords of
section~\ref{sec:prefix-free-mult-binary-alphabet} that instead have
length $K$.  Specifically, they are of the form
$00(1)^{i}0(1)^{K-4-i}0$, where $i\in\{2,\dots,K-6\}$.

\begin{lemma}\label{lem:prefix-delim}
  Given that no strings selected in $W$ have a forbidden prefix, any
  prefix-free $K$-partition of $W$ must select all delimiters.
  Furthermore, the delimiters do not collide with any valid selection
  of the original strings from the set $\mathcal{W}$.
\end{lemma}

\begin{proof}
  By Lemma~\ref{lem:no-split-codewords} the delimiters, which are
  simply longer codewords, cannot contain a cut point in a middle
  position.  Since they are of the maximum string length, $K$, they
  must be entirely selected in $W$.  As a consequence of
  Lemma~\ref{lem:no-split-codewords}, in any valid partition of the
  set $\mathcal{W}$, the selected strings are either: (i) single
  codewords, (ii) a codeword prefixed by a $1$, (iii) a codeword
  suffixed by a $1$, or (iv) two adjacent codewords.  Since each case
  is no longer than a delimiter, it is sufficient to show that none
  could be a prefix of a delimiter.  In all four cases, none could be
  a prefix as the selected string would contain at least four $0$s in
  the first $K/2 + 1$ positions, whereas a delimiter contains at most
  three.\qed
\end{proof}

\paragraph{Construction of forbidden string:}

We now construct the string $F=F_4F_3F_2F_1$.  The string is
constructed in a meticulous manner to ensure that each sub-part must
select forbidden strings from a different partition of the forbidden
set $\mathcal{F}$.  In particular, $F_1=10^{3K-2}111$ and it forces $11$
to be selected; $F_2=0010^{K-3}0^{K-2}11101001^{K-2}01$ and it forces
$101$ and $01$ to be selected; $F_3=001^{K-3}010001$ and it forces
$10001$ to be selected; and $F_4=001^{k-4}010001$ and it forces $0001$
to be selected.

\begin{lemma}\label{lem:prefix-F}
  In any prefix-free partition of $F$ the set of forbidden strings
  must be selected.  Furthermore, there exists a prefix-free partition
  of $F$ that does not collide with a valid partition of $W$ assuming
  that $W$ does not contain a forbidden prefix.
\end{lemma}

\begin{proof} (Sketch).  At a high level, each sub-part of $F$ is
  constructed to ensure that one or more forbidden strings are
  selected, and the remainder of the sub-part consists of $K$ length
  sub-strings that must be entirely selected.  Furthermore, the
  construction of $F$ ensures there is always a cut point between
  sub-parts and sub-part $F_j$ is constructed with the knowledge of
  strings forbidden in $F_i$, for all $i < j$.  The proof requires a
  detailed and exhaustive case analysis.  We give a sketch of the
  correctness.  Consider sub-part $F_1$.  It contains the $3K$-length
  substring $\alpha=10^{3K-2}1$.  At least three strings must be
  selected to cover $\alpha$.  However, it cannot be more than three
  as otherwise at least two contain only $0$ letters and therefore one
  must be a prefix of another.  The only cover for $\alpha$ consisting
  of three strings must select $10^{K-1}$, $0^{K}$, and $0^{K-1}1$,
  regardless of the sub-strings preceding or succeeding $\alpha$.
  Since $1$ cannot be selected alone (without being a prefix of
  $10^{K-1}$) the string $11$ must be selected.  Similar arguments,
  and the fact that $11$ must already be selected, ensures that $F_2$
  is partitioned as $0010^{K-3}$, $0^{K-2}11$, $101$, $001^{K-2}$,
  $01$, thus forbidding $101$ and $01$.  For $F_3$, since $101$ and
  $11$ are already forbidden, there cannot be a cut point prior to any
  $1$ within the left-most run of $1$s.  Since $01$ is forbidden,
  there cannot be a cut point in the second position, nor immediately
  after the left-most run of $1$s.  It follows that $001^{K-3}$ must be
  entirely selected, regardless of the string preceding $F_3$.  The
  remaining string $10001$ must be entirely selected, otherwise it
  would conflict with $101$ or $01$.  Similarly, the partitioning of
  $F_4$ ensures that $001^{k-4}01$ and $0001$ must be selected.

  Note that all strings selected in $F$ not in the forbidden set have
  length $K$.  Thus, to show no new collisions have been introduced,
  it is sufficient to show that no string selected in a valid
  partition of $W$, that does not contain a forbidden prefix, cannot
  be a prefix of one of these strings.  As noted earlier, every
  selected string in a valid partition of $W$ contains at least one
  codeword as a sub-string.  Each codeword contains three runs of
  $0$s; however, each $K$-length selected string in $F$ contains no
  more than two runs of $0$s.\qed
\end{proof}

By our straightforward polynomial time and space reduction of a binary
PF-MSP instance into a binary PF-SP instance and by
Lemmas~\ref{lem:prefix-delim} and ~\ref{lem:prefix-F} and
Theorem~\ref{thm:PF-MSP-NPC-binary} we have the following result.

\begin{theorem}\label{thm:PF-SP-NPC-binary}
  Prefix(Suffix)-Free String Partition (PF-MSP) is \NP-complete for
  binary alphabet ($L=2$).
\end{theorem}
}

\section{Conclusion}
We have established the complexity of the following fundamental
question: given a string $w$ over an alphabet $\Sigma$ and an integer
$K$, can $w$ be partitioned into factors no longer than $K$ such that
no two \emph{collide}?  We have shown this problem is \NP-complete for
versions requiring that no string in the partition is a
copy/factor/prefix/suffix of another.  Furthermore, we have shown the
problems remain hard even for binary strings.  This resolves a number
of open questions from previous work~\cite{ConManTha2008} and
establishes the theoretical hardness of a practical problem in
contemporary synthetic biology, specifically, the oligo design for
gene synthesis problem.

\bibliography{bibliography}

\newpage

\appendix{
\FF 
\PSF }

\end{document}